\documentclass[journal]{IEEEtran}

\usepackage{amsmath}
\usepackage{amssymb}
\usepackage{amsfonts}
\usepackage{graphicx}
\usepackage{epsfig}
\usepackage{subfigure}
\usepackage{psfrag}
\usepackage{color}
\usepackage[noadjust]{cite}

\title{Collaborative Wireless Energy and Information Transfer in Interference Channel
\thanks{The authors are with the Department of Electrical and Computer Engineering, National University of Singapore, Singapore (email: \{s.lee, liu\_liang, elezhang\}@nus.edu.sg). R. Zhang is also with the Institute for Infocomm Research, A*STAR, Singapore.} 
}
\author{Seunghyun Lee,~\IEEEmembership{Student Member,~IEEE}, Liang Liu, and  Rui Zhang,~\IEEEmembership{Member,~IEEE}}

\newtheorem{proposition}{\underline{Proposition}}[section]

\newcommand{\mv}[1]{\mbox{\boldmath{$ #1 $}}}

\def\E{\mathsf{E}}

\def\phi{\varphi}

\def\l{\left}
\def\r{\right}
\def\({\left(}
\def\){\right)}

\def\b0{{\mathbf{0}}}

\def\mC{{\mathbb{C}}}

\def\cC{\mathcal{C}}

\def\cF{\mathcal{F}}

\def\cK{\mathcal{K}}
\def\cL{\mathcal{L}}
\def\cM{\mathcal{M}}
\def\cN{\mathcal{N}}

\def\cP{\mathcal{P}}

\def\cS{\mathcal{S}}

\newcommand{\nn}{\nonumber}

\def\EH{\mathsf{EH}}
\def\ID{\mathsf{ID}}

\def\NC{\mathsf{NC}}

\def\PC{\mathsf{PC}}
\def\FC{\mathsf{FC}}
\def\IA{\mathsf{IA}}

\def\Pmax{P_\mathsf{max}}

\begin{document}
\maketitle \thispagestyle{empty}

\begin{abstract}
This paper studies the \emph{simultaneous wireless information and power transfer} (SWIPT) in a multiuser wireless system, in which distributed transmitters send independent messages to their respective receivers, and at the same time cooperatively transmit wireless power to the receivers via energy beamforming. Accordingly, from the wireless information transmission (WIT) perspective, the system of interest can be modeled as the classic interference channel, while it also can be regarded as a distributed multiple-input multiple-output (MIMO) system for collaborative wireless energy transmission (WET). To enable both information decoding (ID) and energy harvesting (EH) in SWIPT, we adopt the low-complexity \emph{time switching} operation at each receiver to switch between the ID and EH modes over scheduled time. For the hybrid system, we aim to characterize the achievable rate-energy (R-E) trade-offs by various transmitter-side collaboration schemes. Specifically, to facilitate the collaborative energy beamforming, we propose a new \emph{signal splitting} scheme at the transmitters, where each transmit signal is generally split into an \emph{information signal} and an \emph{energy signal} for WIT and WET, respectively. With this new scheme, first, we study the two-user SWIPT system over the fading channel and derive the optimal mode switching rule at the receivers as well as the corresponding transmit signal optimization to achieve various R-E trade-offs. We also compare the R-E performance of our proposed scheme with transmit energy beamforming and signal splitting against two existing schemes with partial or no cooperation of the transmitters. Next, the general case of SWIPT systems with more than two users is studied, for which we propose a practical transmit collaboration scheme by extending the result for the two-user case: we group users into different pairs and apply the cooperation schemes obtained in the two-user case to each paired group. Furthermore, we present a benchmarking scheme based on joint cooperation of all the transmitters inspired by the principle of \emph{interference alignment}, against which the performance of the proposed scheme is compared.
\end{abstract}


\section{Introduction} \label{Section:Introduction}
\PARstart{S}IMULTANEOUS wireless information and power transfer (SWIPT), as an emerging technology by which mobile devices are enabled with both wireless information and energy access at the same time, has recently drawn significant interests. However, the new consideration of dual wireless information transmission (WIT) and wireless energy transmission (WET) imposes various new challenges on wireless system design (see, e.g., \cite{J_ZH:2013, A_ZZH, J_LZC:2013_a, J_LZC:2013_b, J_FO:2012,J_NZDK:2013,J_DPEP:2014, J_PFS:2013, C_XLZ:2013, A_SLXZ, J_HL:2013, A_ZZH_b, A_NLS, J_LZC:2014, J_NLS:2014,C_SLC:2012, C_TKO:2013, J_PC:2013}). Among others, one critical issue for implementing SWIPT is the practical limitation that existing energy harvesting circuits cannot be used to decode and harvest the radio-frequency (RF) signals concurrently \cite{J_ZH:2013, A_ZZH}. To overcome this difficulty, two practical receiver designs have been proposed for SWIPT, namely time switching (TS) and power splitting (PS) \cite{J_ZH:2013, A_ZZH}. With TS, a receiver switches its operation between the two modes of information decoding (ID) and energy harvesting (EH) over time, while with PS, the received signal is split into two streams with one stream used for ID and the other stream for EH. It is worth noting that TS can be practically implemented at a relatively lower cost as compared to PS, since the former requires only signal switches at the receivers whereas the latter needs more costly signal splitters.

Building upon these two practical receiver designs, a handful of research work on SWIPT has been reported recently. In particular, for point-to-point wireless channels, two practical receiver architectures for SWIPT have been proposed in \cite{A_ZZH} with separated or integrated ID and EH circuits at the receiver, based on which the authors characterized various performance  trade-offs in WIT versus WET via the boundary of a so-called \emph{rate-energy (R-E) region}. In \cite{J_LZC:2013_a} and \cite{J_LZC:2013_b}, the authors have investigated the optimal TS and PS schemes, respectively, for fading SWIPT channels. The TS and/or PS schemes have been further studied in wireless relay-assisted communications  \cite{J_FO:2012,J_NZDK:2013,J_DPEP:2014}. An information-theoretic study on the point-to-point SWIPT channel was also given in \cite{J_PFS:2013}.

Furthermore, for the case of wireless broadcast channels, a multiple-input multiple-output (MIMO) SWIPT system has been first studied in \cite{J_ZH:2013}, which optimizes the spatial transmit precoding for achieving various R-E trade-offs for a pair of ID and EH receivers that are either separated or co-located. It is worth noting that in \cite{J_ZH:2013} the rank-one transmit precoding (namely, energy beamforming) was shown to be optimal if only the efficiency of WET is maximized under a sum-power constraint at the multi-antenna transmitter. The work in \cite{J_ZH:2013} has been extended to multiple-input single-output (MISO) SWIPT systems with more than two single-antenna receivers in \cite{C_XLZ:2013} with TS receivers and in \cite{A_SLXZ} with PS receivers, respectively.  Moreover, SWIPT based broadcast systems have been further investigated in multiuser orthogonal frequency division multiplexing (OFDM) channels \cite{J_HL:2013,A_ZZH_b,A_NLS}, and also for secrecy beamforming design problems \cite{J_LZC:2014,J_NLS:2014}.

Besides the point-to-point and point-to-multipoint (i.e., broadcast channel) setups, the study on SWIPT for the more general multipoint-to-multipoint systems has been recently pursued in \cite{C_SLC:2012, C_TKO:2013, J_PC:2013}, in which multiple transmitters send independent messages to their corresponding receivers, and at the same time broadcast power wirelessly to all the receivers. From the perspective of WIT, the system can be modeled as the classic  interference channel (IC), while it also can be regarded as a MIMO WET system with distributed transmitter and receiver nodes. Specifically, in \cite{C_SLC:2012} and \cite{C_TKO:2013}, the authors have studied various transmit beamforming schemes in MISO-IC based SWIPT systems with TS and PS receivers, respectively. Furthermore, in \cite{J_PC:2013}, a two-user SWIPT system under the MIMO-IC setup with TS receivers has been investigated, where the two receivers are assumed to switch among the following four possible operation modes: mode $(\EH,\EH)$, where both receivers harvest energy, mode $(\ID,\EH)$ (or mode $(\EH,\ID)$), where one receiver decodes information (from its intended transmitter) and the other receiver harvests energy, or mode $(\ID,\ID)$, where both receivers decode information from their corresponding transmitters. For each of the above four operation modes, the achievable R-E trade-offs have been analyzed in \cite{J_PC:2013}, especially for the high signal-to-noise ratio (SNR) regime by assuming independent (non-collaborative) WET of the two transmitters.

\begin{figure}
\centering
\includegraphics[width=8cm]{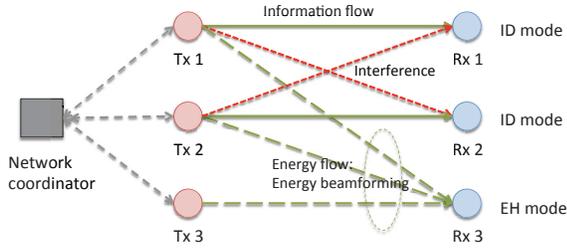} 
\caption{A three-user SWIPT system with collaborative WET.} \label{Fig:IC} \vspace{-15pt}
\end{figure}

Although the SWIPT system under the general multipoint-to-multipoint setup is practically  modeled as MIMO-IC for WIT due to the lack of joint processing and message sharing over the transmitters, from the perspective of WET, we can further improve the energy transfer efficiency of MIMO WET by jointly optimizing the energy signal waveforms at different transmitters based on the MIMO channels to the receivers, thus achieving an \emph{energy beamforming} gain \cite{J_ZH:2013}. However, different from \cite{J_ZH:2013} where the transmit antennas are all equipped at one single transmitter and thus practically subject to a sum-power budget, the energy beamforming design here needs to consider a set of individual power constraints for distributed transmitters.  It is worth pointing out that the energy signals at distributed transmitters can be designed offline and stored for real-time transmission (as will be shown later in this paper), which is in sharp contrast to information signals that are independent and randomly distributed over different transmitters and as a result their real-time joint processing is more difficult to be implemented than collaborative energy beamforming. Furthermore, in this paper we introduce a new \emph{signal splitting} scheme for distributed transmitters, where each transmit signal is in general composed of an energy signal component and an information signal component, in order to facilitate collaborative WET (to other receivers in EH mode) via energy beamforming concurrently with WIT (to its intended receiver in ID mode). It is also assumed that the offline-designed energy signals are perfectly known at all the receivers and thus they can be practically canceled at each receiver prior to decoding the desired information signal. For the purpose of illustration, an example of a three-user SWIPT system with proposed collaborative WET is depicted in Fig.~\ref{Fig:IC}, where a network coordinator is assumed to collect the information from all the transmitters required for the joint design of transmitters' signal splitting and energy beamforming, and then send the designed parameters to individual transmitters for their real-time transmission. It is worth pointing out that the signal splitting concept has also been introduced in other SWIPT systems, but for different purposes; for example, in \cite{J_LZC:2014} and \cite{J_NLS:2014}, the transmit signal is split into information and energy signals, where the latter carries artificial noise for protecting the information sent from being eavesdropped by the energy receivers.

The main results of this paper are summarized as follows:

\begin{itemize}
\item
 First, we consider the special case of a two-user SWIPT system to obtain insight as well as optimal design. In this case, similar to \cite{J_PC:2013}, we assume that the SWIPT system operates by switching among the four modes of $(\EH,\EH)$, $(\ID,\EH)$, $(\EH,\ID)$, and $(\ID,\ID)$. However, different from \cite{J_PC:2013}, due to the newly proposed transmit cooperation with energy beamforming and signal splitting, new analysis is given to characterize the achievable R-E performance. Specifically, for mode $(\EH,\EH)$, we show that energy beamforming with one single energy beam is optimal for maximizing the WET efficiency under per-transmitter power constraints, and also derive the closed-form expression for the optimal energy beamforming weights. We then apply the energy beamforming design jointly with signal splitting to optimize the R-E trade-off for mode $(\ID,\EH)$ or $(\EH,\ID)$. The results obtained are applied to SWIPT over block-fading channels, under which the optimal mode switching rule at the receivers and the corresponding transmit signal optimization are solved. We also compare the R-E performance of our proposed scheme with collaborative energy beamforming and signal splitting against two existing schemes in the literature with partial transmit collaboration \cite{J_PC:2013} or no transmit collaboration \cite{J_LZC:2013_a}. We show that our new scheme achieves substantially enlarged R-E regions as compared to these two baseline schemes.

\item Next, we study the general $K$-user SWIPT system with $K>2$. Due to the prohibitive complexity of exhaustively searching for the optimal operation modes for all users as well as the associated signal splitting and precoding designs for collaborative WET when $K$ becomes large,  we propose a suboptimal scheme of lower complexity. This scheme, called \emph{pairwise cooperation}, is based on users' pairwise collaboration, where we divide the $K$ users into $K/2$ groups (assuming $K$ is an even integer) with each group consisting of two transmitter-receiver pairs. Then over different paired groups, we apply the collaboration schemes obtained for the two-user case. Furthermore, to obtain a performance benchmark, we present a new scheme based on joint collaboration of all the transmitters, which is inspired by a recently introduced transmission technique for the $K$-user IC, called \emph{ergodic interference alignment} \cite{J_NGJV:2012}. In this scheme, all the transmitters/receivers switch their operation modes between ID and EH synchronously, where in ID mode, the transmitters employ ergodic interference alignment for collaborative WIT while in EH mode, they employ energy beamforming (with more than one energy beams in general) for collaborative WET. 

\end{itemize}

The rest of this paper is organized as follows. Section~\ref{Section:SystemModel} introduces the system model of the $K$-user SWIPT system and our proposed signal design. Section~\ref{Section:ProblemFormulation} characterizes the optimal R-E trade-off over fading channels for the two-user case, as compared to two existing schemes. Section~\ref{Section:Multi-user} considers the multiuser SWIPT system. Section~\ref{Section:Numerical} provides simulation results. Finally, Section \ref{Section:Conclusion} concludes the paper.

{\it Notation}: Scalars are denoted by lower-case letters, vectors by bold-face lower-case letters, and matrices by bold-face upper-case letters. $\mv{I}$ and $\mv{0}$  denote an identity matrix and an all-zero matrix, respectively, with appropriate dimensions. For a square matrix $\mv{S}$, ${\rm Tr}(\mv{S})$, ${\rm Rank}(\mv{S})$, and $\mv{S}^{-1}$ denote its trace, rank, and inverse (if $\mv{S}$ is full-rank), respectively; $\mv{S}\succeq\mv{0}~(\mv{S}\preceq\mv{0})$ means that $\mv{S}$ is positive (negative) semi-definite. ${\rm Diag}(\mv{a})$ denotes a diagonal matrix with the main diagonal given by vector $\mv{a}$. For a matrix $\mv{M}$ of arbitrary size, $\mv{M}^{H}$ and $\mv{M}^{T}$ denote the conjugate transpose and transpose of $\mv{M}$, respectively; and $[\mv{M}]_{k,l}$ denotes the $(k,l)$ element of $\mv{M}$. $\E[\cdot]$ denotes the statistical expectation. The distribution of a circularly symmetric complex Gaussian (CSCG) random vector with mean $\mv{x}$ and covariance matrix $\mv{\Sigma}$ is denoted by $\cC\cN(\mv{x},\mv{\Sigma})$, and $\sim$ stands for ``distributed as''. $\mathbb{C}^{x \times y}$ denotes the space of $x\times y$ complex matrices. $\|\mv{z}\|$ denotes the Euclidean norm of a complex vector $\mv{z}$, while $|z|$ and $z^*$ are the absolute value and the complex conjugate of a complex number $z$, respectively.

\section{System Model} \label{Section:SystemModel}
We consider a $K$-user SWIPT system consisting of $K$ single-antenna transmitter-receiver (Tx-Rx) pairs, denoted by the set $\cK = \{1,2,...,K\}$. It is assumed that all Txs share the same band for WIT and WET over flat-fading channels. For WIT, the system can be modeled by the $K$-user single-input single-output (SISO) IC, since we do not consider joint information processing at different Txs. The baseband complex channel coefficient from Tx $l$ to Rx $k$ $(k,l\in\cK)$ is denoted by $h_{kl}$. It is assumed that $h_{kl}$'s are all known at the network coordinator (see Fig.~\ref{Fig:IC}), which provides the transmit design for all Txs. For convenience, we define the channels from all Txs to Rx $k$ in a vector $\mv{h}_k = [h_{k1},...,h_{kK}]$, $k\in\cK$. The discrete-time signal received at Rx $k$ is then given by
\begin{equation*} 
y_k(n) = \mv{h}_k\mv{x}(n) + z_k(n), \; k\in\cK, 
\end{equation*}
where $n$ denotes the symbol index; $\mv{x}(n)=[x_1(n),...,x_K(n)]^T$ is the transmit signal vector with $x_k(n)$ denoting the transmitted signal from Tx $k$; and $z_k(n)$ is the additive noise at Rx $k$. It is assumed that $z_k(n)\sim\cC\cN(0,\sigma_k^2)$, $\forall k\in\cK$. We consider the practical \emph{peak-power constraint} at each Tx, which limits the instantaneous transmit power, i.e., $\E_n[|x_k(n)|^2]\leq \Pmax$, $\forall k\in\cK$, where $\Pmax$ denotes the maximum power budget at all Txs.

In order to enable collaborative WET along with WIT, in this paper we employ a new  \emph{signal splitting} scheme, by which the transmit signal at each Tx in general consists of an  information signal component and an energy signal component, i.e., 
\begin{equation*}  
x_k(n) = x_k^I(n) + x_k^E(n), \; k\in\cK,
\end{equation*}
where $x_k^I(n)$ and $x_k^E(n)$ denote the information signal and energy signal at Tx $k$, respectively. It is assumed that the information signal $x_k^I(n)$ is an independent and identically distributed (i.i.d.) CSCG random variable with zero-mean and variance (power) $p_k^I$, denoted by $x_k^I(n)\sim\cC\cN(0,p_k^I)$, $k\in\cK$. Furthermore, for the energy signal $x_k^E(n)$, since it does not carry any information, it can be designed as a zero-mean pseudo-random signal with arbitrary distribution, provided that its power spectral density satisfies certain regulations on radio signal radiation for the operating band of interest. In this paper, we assume that $x_k^E(n)$ is also a CSCG random signal, i.e., $x_k^E(n)\sim\cC\cN(0,p_k^E)$, $k\in\cK$, with $p_k^E$ denoting its average power. Note that given the peak-power constraint $\Pmax$, we have $p_k^I+p_k^E\leq\Pmax$, $\forall k\in\cK$. 

For convenience, we define the transmit covariance matrix for the energy signals from all $K$ Txs as $\mv{S}_E = \E_n[\mv{x}_E(n)\mv{x}_E^H(n)]$, where $\mv{x}_E(n) = \l[x_1^E(n),...,x_K^E(n)\r]^T$. In practice, $\mv{S}_E$ conveys all the required parameters (i.e., power allocations and beamforming weights) for the design of collaborative WET by all Txs. Let ${\rm Rank}(\mv{S}_E)=d_E$, with $1\leq d_E\leq K$, and the eigenvalue decomposition of $\mv{S}_E$ be denoted by $\mv{S}_E=\mv{V}\mv{\Sigma}\mv{V}^H$, where $\mv{V}\in\mC^{K\times d_E}$, $\mv{V}^H\mv{V}=\mv{I}$, is the precoding matrix, and $\mv{\Sigma}={\rm Diag}(q_1,...,q_{d_E})$ with $q_1,...,q_{d_E}$ denoting the positive eigenvalues of $\mv{S}_E$. Then we can express the energy signal vector $\mv{x}_E(n) = \sum_{i=1}^{d_E}\sqrt{q_i}\mv{v}_i s_i^E(n)$, where $\mv{v}_i$ is a beamforming vector, which is drawn from $\mv{V}=[\mv{v}_1,...,\mv{v}_{d_E}]$, and $s_1^E(n),...,s_{d_E}^E(n)$ are i.i.d. pseudo-random  variables with $s_i^E(n)\sim\cC\cN(0,1)$, $i=1,...,d_E$. Note that for the special case of $d_E=1$, one single energy beam is used and all Txs transmit the same pre-determined energy signal $s_1^E(n)$ with different weights drawn from $\sqrt{q_1}\mv{v}_1$. From a practical consideration, it is desirable to have small $d_E$ due to the following two reasons. First, it is practically sensible to reduce the number of energy signals stored at each Tx which is equal to $d_E$. Second, as will be shown later in this paper, the pre-designed energy signals should be canceled at each Rx prior to ID to improve the WIT rate, which requires that each Rx cancels up to $d_E$ number of interference signals due to WET; thus, it is desirable to keep $d_E$ small.  

In this paper, we adopt the TS scheme at each Rx. For convenience, we define an indicator function to denote the working mode of Rx $k$ as follows:
\begin{align} \label{Eq:Mode}
\rho_k = \l\{\begin{aligned}
&1, \quad \mbox{ID mode is active at Rx $k$}, \\
&0, \quad \mbox{EH mode is active at Rx $k$}.
\end{aligned}
\r.
\end{align}

Using \eqref{Eq:Mode}, for WET, the harvested power at Rx $k$ due to the information/energy signals from all Txs can be expressed as 
\begin{align} 
Q_k & =  \zeta(1-\rho_k)\E_n[|\mv{h}_k\mv{x}(n)|^2] \nn\\
& =\zeta(1-\rho_k)\l(\sum_{l\in\cK} |h_{kl}|^2 p_l^I +  \mv{h}_k\mv{S}_E\mv{h}_k^H\r), \; k\in\cK, \label{Eq:Energy2}
\end{align}
where the constant $0<\zeta \leq 1$ represents the efficiency in harvesting and storing received energy. For notational brevity, we assume $\zeta=1$ in the sequel, unless otherwise stated. Note that since the background noise power $\sigma_k^2$ is practically much smaller as compared to the average received signal power from the viewpoint of WET, here we have ignored it in the expression of harvested power. On the other hand, for WIT, it is assumed that the interference at Rx $k$ due to the energy signals, i.e., $x_l^E(n)$, $l\in\cK$, can be first perfectly canceled by Rx $k$, since the  energy signals are pre-designed pseudo-random signals which can be stored at all Rxs for interference cancellation. However, the interference due to the information signals from other Txs, i.e., $x_l^I(n)$, $l\in\cK$, $l\neq k$, remains and is assumed to be additional noise at each Rx $k$, for a practical receiver implementation.  Therefore, for WIT, the achievable rate at Rx $k$ can be expressed as
\begin{equation}\label{Eq:Rate}
R_k = \rho_k\log_2\l(1+\frac{|h_{kk}|^2 p_k^I}{\sum_{l\in\cK,l\neq k}|h_{kl}|^2p_l^I + \sigma_k^2}\r), \; k\in\cK.
\end{equation} 

In the following two sections, we first investigate the transmit collaboration designs and the achievable R-E performance for the special case of a two-user SWIPT system over fading channels, and then address the general $K$-user SWIPT system with $K>2$.

\section{Collaborative Transmission for SWIPT: Two-user Case} \label{Section:ProblemFormulation}
In this section, we focus on the two-user SWIPT system, i.e., $K=2$, over flat-fading channels. For the purpose of exposition, in the sequel we use index $\nu$ to indicate channel fading state, e.g., $h_{kl}(\nu)$ denotes the channel from Tx $l$ to Rx $k$ at fading state $\nu$. We assume the block fading model such that the channel $h_{kl}(\nu)$, $k,l\in\cK$, remains constant during each block for a given fading state $\nu$, but can vary from block to block as $\nu$ changes. Furthermore, we define $\mv{\rho}(\nu)=(\rho_1(\nu), \rho_2(\nu))$ as the working modes of the two users at fading state $\nu$, and denote the set of all four possible mode combinations as $\cM = \{\mv{\rho}(\nu):\rho_k(\nu) \in \{0,1\}, k=1,2\}$. Specifically, the four modes are $(\EH,\EH)$ with $\mv{\rho}(\nu)=(0,0)$, $(\ID,\EH)$ with $\mv{\rho}(\nu)=(1,0)$, $(\EH,\ID)$ with $\mv{\rho}(\nu)=(0,1)$, and $(\ID,\ID)$ with $\mv{\rho}(\nu)=(1,1)$, at fading state $\nu$, similar to those considered in \cite{J_PC:2013}.

In the rest of this section, we first formulate the design problem for characterizing the optimal R-E trade-off of the two-user system, by jointly optimizing the Rxs' mode switching rule and Txs' collaborative signal design. Next, we derive the optimal solution to this problem. Finally, we introduce two suboptimal schemes based on the existing results in \cite{J_PC:2013, J_LZC:2013_a}.

\subsection{Problem Formulation}
In this paper, we consider two performance metrics for the SWIPT system, which are the average sum-capacity for WIT and the average harvested power of individual Rxs for WET. For convenience, we define $\mv{p}_I(\nu)=(p_1^I(\nu),p_2^I(\nu))$ as the power allocation vector to the information signals for the two Txs at fading state $\nu$. It is worth noting that at one particular fading state $\nu$, the harvested power and achievable rate given in \eqref{Eq:Energy2} and \eqref{Eq:Rate}, respectively, are functions of $\mv{\rho}(\nu)$, $\mv{p}_I(\nu)$, and/or  $\mv{S}_E(\nu)$.  To characterize the optimal R-E trade-off, we formulate the following problem by jointly optimizing $\mv{\rho}(\nu)$, $\mv{p}_I(\nu)$, and $\mv{S}_E(\nu)$.
\begin{align}
\mathrm{(P1)}:~\mathop{\mathtt{Maximize}}_{\{\mv{\rho}(\nu),\mv{p}_I(\nu),\mv{S}_E(\nu)\}} &~~  \E_\nu\l[R_1(\nu) + R_2(\nu)\r] \\
\mathtt{subject\; to}\quad&~~  \E_\nu\l[Q_k(\nu)\r] \geq \bar{Q}_k \label{Ineq:P1}, \;k=1,2\\
&~~\mv{\rho}(\nu) \in \cM , \;\forall \nu \label{Constraint:P1}\\
&~~\{\mv{p}_I(\nu),\mv{S}_E(\nu)\} \in \cF, \;\forall \nu, \nn
\end{align}
where $\cF$ is the feasible set for $\{\mv{p}_I(\nu),\mv{S}_E(\nu)\}$, defined as
\begin{align} \label{Eq:Feasible}
\cF &= \{\mv{p}_I(\nu),\mv{S}_E(\nu): \mv{S}_E(\nu)\succeq\mv{0}, p_k^I(\nu)\geq 0, \nn\\
&[\mv{S}_E(\nu)]_{k,k}=p_k^E(\nu) \geq 0, p_k^I(\nu)+p_k^E(\nu) \leq\Pmax, k=1,2\},
\end{align}
and $\bar{Q}_k$ is the average harvested power requirement for Rx $k$. By solving problem (P1), the network coordinator obtains the optimal operation modes $\mv{\rho}(\nu)$ for the two users as well as the corresponding optimal transmit power allocation $\mv{p}_I(\nu)$ and energy beamforming matrix $\mv{S}_E(\nu)$ at the two Txs at each fading state $\nu$. We refer to this cooperation scheme for SWIPT as \emph{full cooperation} (FC).

It is worth noting that problem (P1) is in general non-convex, since the objective function is non-concave over $\mv{p}_I(\nu)$, and furthermore the constraints in \eqref{Ineq:P1} and \eqref{Constraint:P1} are in general non-convex due to the binary variables for mode switching.  However, under the assumption that the fading channel distribution is continuous over $\nu$, it can be shown that   strong duality still approximately holds for (P1), since this problem satisfies the so-called \emph{time-sharing} condition \cite{J_YL:2006}. As a result, we can apply the Lagrange duality method to solve (P1) optimally, for which the detail is given next.

\subsection{Optimal Solution} \label{Section:OptimalSolution_Cooperative}
In this subsection, we study the optimal solution of problem (P1) with the FC scheme. First, the Lagrangian of (P1) is formulated as
\begin{align*} 
&\cL^\FC(\mv{\rho}(\nu),\mv{p}_I(\nu),\mv{S}_E(\nu),\mu_1,\mu_2) = \E_\nu\l[R_1(\nu) + R_2(\nu)\r]  +\nn\\ & \quad\mu_1\l(\E_\nu\l[Q_1(\nu)\r]-\bar{Q}_1\r)+\mu_2\l(\E_\nu\l[Q_2(\nu)\r]-\bar{Q}_2\r),
\end{align*}
where $\mu_1,\mu_2 \geq 0$ are the dual variables associated with the constraints in \eqref{Ineq:P1} for $k=1,2$, respectively. Then, the Lagrange dual function of (P1) is given by
\begin{align} \label{Dual Problem}
&g^\FC(\mu_1,\mu_2) = \nn\\ 
&\max_{\mv{\rho}(\nu)\in \cM,  \{\mv{p}_I(\nu),\mv{S}_E(\nu)\}\in\cF}\cL^\FC(\mv{\rho}(\nu),\mv{p}_I(\nu),\mv{S}_E(\nu),\mu_1,\mu_2).
\end{align}
The resulting dual problem of (P1) is thus given as follows.
\begin{align*}
\mathrm{(D1)}:~\mathop{\mathtt{Maximize}}_{\mu_1,\mu_2} &~~ g^\FC(\mu_1,\mu_2)   \\
\mathtt{subject \; to}&~~  \mu_1\geq 0,\mu_2 \geq 0.  
\end{align*}

The maximization problem in \eqref{Dual Problem} is for obtaining the dual function, which can be efficiently solved by considering a set of subproblems all having the same structure and each for one particular fading state $\nu$. For one particular fading state $\nu$, the associated subproblem is expressed as
\begin{equation} \label{Subproblem}
\max_{\mv{\rho}\in \cM,\{\mv{p}_I,\mv{S}_E\}\in\cF} f_\nu^\FC(\mv{\rho},\mv{p}_I,\mv{S}_E),
\end{equation}
where by discarding some irrelevant constant terms in $\cL^\FC(\cdot)$, we have
\begin{equation} \label{Eq:f_FC}
f_\nu^\FC(\mv{\rho},\mv{p}_I,\mv{S}_E)= R_1 + R_2 + \mu_1 Q_1 + \mu_2 Q_2.
\end{equation}
Note that the fading state index $\nu$ has been omitted in the above formulation for brevity. Problem \eqref{Dual Problem} can thus be solved by solving parallel problems in \eqref{Subproblem} for different  fading states, given $\mu_1$ and $\mu_2$. It is then observed that problem \eqref{Subproblem} for each fading state $\nu$ can be solved by first finding the optimal solution, denoted by $\bar{\mv{p}}_I$ and $\bar{\mv{S}}_E$, to maximize $f_\nu^\FC(\mv{\rho},\mv{p}_I, \mv{S}_E)$ in \eqref{Eq:f_FC} with each given $\mv{\rho}\in\cM$, and then by comparing the resulting values of $f_\nu^\FC(\mv{\rho},\bar{\mv{p}}_I,\bar{\mv{S}}_E)$ over $\mv{\rho}\in\cM$ to obtain the optimal solution for $\mv{\rho}$, denoted by $\mv{\rho}^\star$, i.e., $\mv{\rho}^\star=\arg \max_{\mv{\rho}\in M}f_\nu^{\FC}(\mv{\rho},\bar{\mv{p}}_I,\bar{\mv{S}}_E)$. Finally, with the obtained $\mv{\rho}^\star$, the corresponding optimal solution for $\mv{p}_I$ and $\mv{S}_E$ of problem \eqref{Subproblem} can be found, denoted by $\mv{p}_I^\star$ and $\mv{S}_E^\star$, respectively. In the following, we solve problem \eqref{Subproblem} for different modes of $\mv{\rho}\in\cM$.

\subsubsection{Mode $(\EH,\EH)$} \label{Subsection:a}
Consider first the case of $\mv{\rho} = (0,0)$. According to \eqref{Eq:Rate}, we have $R_1=R_2=0$. Note that for this case, we can easily have $\bar{p}_1^I=\bar{p}_2^I=0$, since the two Txs do not send information. It thus follows from \eqref{Eq:Energy2} that  $Q_1 = \mv{h}_1\mv{S}_E\mv{h}_1^H$ and $Q_2 = \mv{h}_2\mv{S}_E\mv{h}_2^H$. It can then be shown that problem \eqref{Subproblem} in this case is equivalent to the following problem.
\begin{align}
\mathrm{(P1.1)}:~\mathop{\mathtt{Maximize}}_{\mv{S}_E} &~~  \mu_1 \mv{h}_1\mv{S}_E\mv{h}_1^H + \mu_2 \mv{h}_2\mv{S}_E\mv{h}_2^H \\
\mathtt{subject\; to}&~~  {\rm Tr}(\mv{I}_1\mv{S}_E) \leq \Pmax \label{Ineq:P3.1_1}\\
&~~  {\rm Tr}(\mv{I}_2\mv{S}_E) \leq \Pmax \label{Ineq:P3.1_2}\\
&~~ \mv{S}_E\succeq\mv{0}, \nn
\end{align}
where $\mv{I}_1$ and $\mv{I}_2$ are defined as 
\begin{equation*}
\mv{I}_1=\l[\begin{array}{cc} 1 & 0 \\ 0 & 0  \end{array}\r], \quad
\mv{I}_2=\l[\begin{array}{cc} 0 & 0 \\ 0 & 1  \end{array}\r].
\end{equation*}
It can be shown that (P1.1) is a semidefinite program (SDP) and thus can be solved efficiently via existing software, e.g., CVX \cite{CVX}. However, in the following proposition we present a closed-form solution to (P1.1) to provide further insight. 
\begin{proposition} \label{Proposition:P1}
The optimal solution to (P1.1), denoted by $\bar{\mv{S}}_E$, is given by
\begin{equation} \label{Eq:OptS}
\bar{\mv{S}}_E=\Pmax\l[\begin{array}{cc} 1 & \alpha  \\ \frac{1}{\alpha}  & 1  \end{array}\r],
\end{equation}
where $\alpha = \tilde{\mv{h}}_1^H \tilde{\mv{h}}_2/\l|\tilde{\mv{h}}_1^H \tilde{\mv{h}}_2\r| = \l|\tilde{\mv{h}}_1^H \tilde{\mv{h}}_2\r|/\tilde{\mv{h}}_2^H \tilde{\mv{h}}_1$, with $\tilde{\mv{h}}_{k} = [\sqrt{\mu_1} h_{1k}, \sqrt{\mu_2} h_{2k}]^T$, $k=1,2$.
\end{proposition}
\begin{proof}
See Appendix~\ref{Proof:Proposition:P1}.
 \end{proof}
 
From \eqref{Eq:OptS}, it is observed that the two Txs should both transmit with maximum power $\Pmax$, and furthermore the optimal transmit covariance can be expressed as  $\bar{\mv{S}}_E = \Pmax\mv{v}\mv{v}^H$ where $\mv{v}=[1,1/\alpha]^T$ is the  beamforming vector, since we have $\alpha^*=1/\alpha$ and $|\alpha|=1$. In other words, $\bar{\mv{S}}_E$ is of rank-one, i.e., only one single energy beam is used for collaborative energy beamforming at the two Txs. As a result, Tx $1$ and Tx $2$ only need to store one common pseudo-random energy signal and the network coordinator only needs to send the phase of $\alpha$ (a real scalar between $0$ and $2\pi$) to Tx $2$ to implement the optimal collaborative energy beamforming.

\subsubsection{Mode $(\ID,\EH)$ or $(\EH,\ID)$} \label{Subsection:b}
Next, consider the case of $\mv{\rho} = (1,0)$, where similar results can be obtained for the case of $\mv{\rho}=(0,1)$ and thus are omitted. For mode $(\ID,\EH)$ with $\mv{\rho}=(1,0)$, according to \eqref{Eq:Energy2} and \eqref{Eq:Rate}, we have $Q_1=0$ and $R_2=0$. Note that for this case, it easily follows that $\bar{p}_2^I=0$, since Tx $2$ does not transmit information. As a result, according to \eqref{Eq:Energy2} and \eqref{Eq:Rate}, we have $R_1 = \log_2\l(1+\frac{|h_{11}|^2 p_1^I}{\sigma_1^2}\r) $ and $Q_2 =|h_{21}|^2p_1^I +  \mv{h}_2\mv{S}_E\mv{h}_2^H$. Then problem \eqref{Subproblem} in this case can be expressed as
\begin{align}
\label{P1.2}\mathrm{(P1.2)}:~\mathop{\mathtt{Maximize}}_{\{p_1^I,\mv{S}_E\}} &~~  \log_2\l(1+\frac{|h_{11}|^2 p_1^I}{\sigma_1^2}\r) +\nn\\& \mu_2\l(|h_{21}|^2p_1^I +  \mv{h}_2\mv{S}_E\mv{h}_2^H\r) \\
\mathtt{subject\; to}&~~\{p_1^I,\mv{S}_E\}\in\cF. \nn
\end{align}
To solve (P1.2), first it can be shown that $\bar{p}_1^I + \bar{p}_1^E = \Pmax$  and $\bar{p}_2^E=\Pmax$ should hold for (P1.2), where from \eqref{Eq:Feasible} we have $\bar{p}_k^E=[\bar{\mv{S}}_E]_{k,k}$, $k=1,2$. In other words, the two Txs should both transmit with maximum power $\Pmax$, since the energy signals from both Txs can be canceled at Rx $1$ and thus it is desirable for the two Txs to transmit their maximum power. With $\bar{p}_1^E = \Pmax-\bar{p}_1^I$ and $\bar{p}_2^E=\Pmax$ at hand, it can be shown that (P1.2) is a special case of (P1.1) with $\mu_1=0$. With $\mu_1=0$ in Proposition~\ref{Proposition:P1}, the optimal $\mv{S}_E$ for (P1.2) can be expressed as $\bar{\mv{S}}_E = \mv{u}\mv{u}^H$, where
\begin{equation} \label{Eq:EB_P1.2}
\mv{u} = \l[\sqrt{\Pmax-p_1^I},\sqrt{\Pmax}\frac{h_{21}^*h_{22}}{|h_{21}^*h_{22}|}\r]^T.
\end{equation}
 To determine the optimal $p_1^I$, i.e., $\bar{p}_1^I$, we substitute $\bar{\mv{S}}_E = \mv{u}\mv{u}^H$ into \eqref{P1.2}, and then (P1.2) reduces to the following problem.
\begin{align}
\label{P1.2'}~\mathop{\mathtt{Maximize}}_{p_1^I} &~~ y(p_1^I)   \\
\mathtt{subject \; to}&~~  0 \leq p_1^I \leq \Pmax,  \label{Ineq:P4.1} \nn
\end{align}
where from \eqref{P1.2} $y(p_1^I)$ is defined as
{\small \begin{equation*}
 y(p_1^I) = \log_2\l(1+\frac{|h_{11}|^2 p_1^I}{\sigma_1^2}\r) +  \qquad\qquad\qquad\qquad\qquad\qquad
\end{equation*} 
\begin{equation*}
 \mu_2\l( |h_{21}|^2 \Pmax + |h_{22}|^2 \Pmax + 2|h_{21}^* h_{22}|\sqrt{(\Pmax-p_1^I)\Pmax}\r).
\end{equation*} }
It can be shown that $y(p_1^I)$ is a concave function of  $p_1^I$ for $0\leq p_1^I\leq\Pmax$; hence, the optimal solution to problem \eqref{P1.2'} can be efficiently obtained by e.g., Newton's method \cite{B_BV:2004}. Thus, $\bar{\mv{S}}_E$ is obtained.

Since $\bar{\mv{S}}_E$ is of rank-one in this case, similar to the case of mode $(\EH,\EH)$, only one energy beam is needed for collaborative energy beamforming at the two Txs. It is worth noting that in this case, according to \eqref{Eq:EB_P1.2}, the network coordinator needs to send the optimal power allocation for the transmitted information signal to Tx $1$, i,e., $\bar{p}_1^I$ by solving problem \eqref{P1.2'}, and the phase of $h_{21}^*h_{22}$ to Tx $2$ to implement collaborative energy beamforming for WET to Rx $2$.

\subsubsection{Mode $(\ID,\ID)$}  \label{Subsection:d}
Finally, consider the case of $\mv{\rho} = (1,1)$. According to \eqref{Eq:Energy2}, we have $Q_1=Q_2=0$. Note that in this case we can easily have $\bar{p}_1^E=\bar{p}_2^E=0$, and thus $\bar{\mv{S}}_E=\mv{0}$, since the two Txs do not transmit energy signals. It thus follows from \eqref{Eq:Rate} that $R_1= \log_2\l(1+\frac{|h_{11}|^2 p_1^I}{|h_{12}|^2 p_2^I + \sigma_1^2}\r)$ and $R_2= \log_2\l(1+\frac{|h_{22}|^2 p_2^I}{|h_{21}|^2 p_1^I + \sigma_2^2}\r)$, and hence  problem \eqref{Subproblem} reduces to 
\begin{align*}
\mathrm{(P1.3)}:~\mathop{\mathtt{Maximize}}_{\mv{p}_I} &~~    \log_2\l(1+\frac{|h_{11}|^2 p_1^I}{|h_{12}|^2 p_2^I + \sigma_1^2}\r) +\nn\\& \log_2\l(1+\frac{|h_{22}|^2 p_2^I}{|h_{21}|^2 p_1^I + \sigma_2^2}\r) \\
\mathtt{subject\; to}&~~0\leq p_k^I \leq \Pmax,\quad k=1,2. 
\end{align*}
(P1.3) is a non-convex problem. However, it has been shown in \cite{J_GGOK:2008} that  on-off power control is optimal for this problem. Specifically, by defining $\cP^\star = \{(0,\Pmax), (\Pmax,0), (\Pmax, \Pmax) \}$, then the optimal solution to (P1.3), denoted by $\bar{\mv{p}}_I$, is given by
\begin{align} \label{Eq:Optimal_P1.1}
\bar{\mv{p}}_I =& \arg\max_{\mv{p}_I\in\cP^\star} \log_2\l(1+\frac{|h_{11}|^2 p_1^I}{|h_{12}|^2 p_2^I + \sigma_1^2}\r) + \nn\\ & \qquad\quad\log_2\l(1+\frac{|h_{22}|^2 p_2^I}{|h_{21}|^2 p_1^I + \sigma_2^2}\r).
\end{align}

It is worth noting that, as can be observed from \eqref{Eq:Optimal_P1.1}, the network coordinator only needs to send an on/off (binary) signal to each Tx in this case, according to the optimal power solution for \eqref{Eq:Optimal_P1.1}. 

To summarize, with a given pair of $\mu_1$ and $\mu_2$, problem \eqref{Subproblem} has been efficiently solved for different operation modes of $\mv{\rho}\in\cM$. Then, problem \eqref{Subproblem} is solved for each fading state $\nu$, by finding the mode $\mv{\rho}$ that maximizes $f^\FC(\mv{\rho},\bar{\mv{p}}_I, \bar{\mv{S}}_E)$ defined in \eqref{Eq:f_FC}. Then, the sub-gradient based method such as ellipsoid method \cite{B_BV:2004} can be applied to iteratively search for the optimal dual solution, defined by $\mu_1^\star$ and $\mu_2^\star$, for problem (D1). The sub-gradient for updating $(\mu_1,\mu_2)$ can be shown to be $(\E_\nu[Q_1(\mv{\rho}^\star(\nu),\mv{p}_I^\star(\nu),\mv{S}_E^\star(\nu))]-\bar{Q}_1, \E_\nu[Q_2(\mv{\rho}^\star(\nu),\mv{p}_I^\star(\nu),\mv{S}_E^\star(\nu))]-\bar{Q}_2)$. Thus, (P1) is solved completely.

\subsection{Suboptimal Schemes} \label{Section:Benchmark}
In this subsection, we introduce two suboptimal solutions to problem (P1) based on existing schemes in \cite{J_PC:2013} and \cite{J_LZC:2013_a}, namely \emph{partial cooperation} and \emph{no cooperation}, respectively, for comparison with our proposed FC scheme.

\subsubsection{Partial Cooperation} \label{Subsection:Partially-Cooperative}
In this scheme, there is no signal splitting applied at each Tx and the transmitted signal at each Tx is only information signal, i.e., $x_k(n) = x_k^I(n)$, $k=1,2$. As a result, collaborative energy beamforming cannot be applied, where the two Txs cooperate by only jointly determining the power allocation (i.e., $\mv{p}_I(\nu)=(p_1^I(\nu),p_2^I(\nu))$) and the Rx operation modes (i.e., $\mv{\rho}(\nu)=(\rho_1(\nu),\rho_2(\nu))$) at each fading state $\nu$. From \eqref{Eq:Energy2}, the harvested power at Rx $k$ at fading state $\nu$ is thus given by
\begin{align} \label{Eq:EnergyC}
Q_k(\nu) = (1-\rho_k(\nu))(|h_{kk}(\nu)|^2 p^I_k(\nu) + &|h_{k\bar{k}}(\nu)|^2 p_{\bar{k}}^I(\nu)), \nn\\ & k=1,2,
\end{align}
where $\bar{k}:=\{1,2\}\backslash\{k\}$. Next, from \eqref{Eq:Rate}, the achievable rate at Rx $k$ at fading state $\nu$ is given by

\begin{equation}\label{Eq:RateC}
R_k(\nu) = \rho_k(\nu)\log_2\l(1+\frac{|h_{kk}(\nu)|^2 p_k^I(\nu)}{|h_{k\bar{k}}(\nu)|^2 p_{\bar{k}}^I(\nu) + \sigma_k^2}\r), \;k = 1,2.
\end{equation} 
Since we have $p_1^E=p_2^E=0$ in this scheme, $\mv{S}_E=\mv{0}$ and thus \eqref{Eq:Feasible} is simplified as 
\begin{equation*}
\cP = \{\mv{p}_I(\nu): 0\leq p_k^I(\nu)\leq \Pmax, k=1,2\}.
\end{equation*}
It then follows that problem (P1) is reduced to the following problem in the case of partial cooperation (PC).
\begin{align}
\mathrm{(P2)}:~\mathop{\mathtt{Maximize}}_{\{\mv{\rho}(\nu),\mv{p}_I(\nu)\}} &~~  \E_\nu\l[R_1(\nu) + R_2(\nu)\r] \nn \\
\mathtt{subject\; to}&~~  \E_\nu\l[Q_k(\nu)\r] \geq \bar{Q}_k , \;k=1,2 \label{Ineq:P2}\\
&~~\mv{\rho}(\nu) \in \cM , \;\forall \nu \nn\\
&~~\mv{p}_I(\nu)\in\cP, \;\forall \nu. \nn
\end{align}

Similar to (P1), problem (P2) can be decoupled into parallel subproblems each for one fading state $\nu$ and expressed as (by omitting the fading state $\nu$)
\begin{equation} \label{Eq:Subproblem2}
\max_{\mv{\rho}\in \cM,\mv{p}_I\in\cP} f_\nu^\PC(\mv{\rho},\mv{p}_I),
\end{equation}
where 
\begin{equation} \label{Eq:f_PC}
f_\nu^\PC(\mv{\rho},\mv{p}_I)= R_1 + R_2 + \mu_1 Q_1 + \mu_2 Q_2, 
\end{equation}
with $\mu_1,\mu_2\geq 0$ denoting the dual variables associated with the constraints in \eqref{Ineq:P2} for $k=1,2$, respectively. Problem \eqref{Eq:Subproblem2} can then be solved by first finding the optimal solution, denoted by $\bar{\mv{p}}_I$, that maximizes $f_\nu^\PC(\mv{\rho},\mv{p}_I)$ in \eqref{Eq:f_PC} with given $\mv{\rho}\in\cM$, and then searching $\mv{\rho}$ that maximizes $f_\nu^\PC(\mv{\rho},\bar{\mv{p}}_I)$ over $\mv{\rho}\in\cM$. Similar to problem \eqref{Subproblem}, (P2) is then solved by searching the optimal dual solution $(\mu_1^\star,\mu_2^\star)$ based on the ellipsoid method. Therefore, in the following we focus on solving \eqref{Eq:Subproblem2} with given $\mv{\rho}\in\cM$. 

\begin{itemize}
\item Mode $(\EH,\EH)$: In this case, $\mv{\rho}=(0,0)$. According to \eqref{Eq:EnergyC} and \eqref{Eq:RateC}, problem \eqref{Eq:Subproblem2} in this case is expressed as
\begin{align*}
\mathrm{(P2.1)}:~\mathop{\mathtt{Maximize}}_{\mv{p}_I} &~~   \mu_1(|h_{11}|^2p_1^I+|h_{12}|^2p_2^I) +\nn\\ &  \mu_2(|h_{21}|^2p_1^I+|h_{22}|^2p_2^I) \\
\mathtt{subject\; to}&~~\mv{p}_I\in\cP.
\end{align*}
It can be observed that the optimal solution of (P2.1) is given by $\bar{\mv{p}}_I = (\Pmax,\Pmax)$.

\item Mode $(\ID,\EH)$: In this case, $\mv{\rho}=(1,0)$. Similar analysis can be made for mode $(\EH,\ID)$ with $\mv{\rho}=(0,1)$, and thus is omitted. According to \eqref{Eq:EnergyC} and \eqref{Eq:RateC},  problem \eqref{Eq:Subproblem2} in this case is expressed as
\begin{align}
\label{P2.2}\mathrm{(P2.2)}:~\mathop{\mathtt{Maximize}}_{\mv{p}_I} &~~  \log_2\l(1+\frac{|h_{11}|^2 p_1^I}{|h_{12}|^2 p_2^I + \sigma_1^2}\r) + \nn\\ &\mu_2 (|h_{21}|^2 p_1^I + |h_{22}|^2 p_2^I) \\
\mathtt{subject\; to}&~~\mv{p}_I\in\cP. \nn
\end{align}
Note that the objective function of the above problem is non-concave over $p_1^I$ and $p_2^I$; thus, problem (P2.2) is not convex. However, this problem can be efficiently solved as follows. First, it can be observed that \eqref{P2.2} monotonically increases with $p_1^I$; thus, we obtain $\bar{p}_1^I = \Pmax$ for problem (P2.2). Next, with $\bar{p}_1^I = \Pmax$, it can be shown that \eqref{P2.2} is a convex function over $0\leq p_2^I \leq \Pmax$; thus, the optimal solution of $p_2^I$, i.e., $\bar{p}_2^I$, is either $0$ or $\Pmax$, from which we simply select the one resulting in the larger function value of \eqref{P2.2}. To summarize, the optimal solution to problem (P2.2) is in the set, $\bar{\mv{p}}_I \in \{(\Pmax, 0), (\Pmax, \Pmax)\}$.

\item Mode $(\ID,\ID)$: Finally, consider the case of $\mv{\rho}=(1,1)$. According to \eqref{Eq:EnergyC} and \eqref{Eq:RateC}, it can be shown that problem \eqref{Eq:Subproblem2} in this case reduces to (P1.3), for which the same on/off solution given in \eqref{Eq:Optimal_P1.1} applies.

\end{itemize}

Based on the above results, it can be inferred that in this case, for each mode, the network coordinator only needs to send an on/off control signal to each of the two Txs since if any of them is switched on, it should transmit with maximum power $\Pmax$.

\subsubsection{No Cooperation}
For another benchmark scheme, we consider the case when there is no cooperation at the two Txs, and as a result the Rxs perform mode switching independently based on their own observed CSI, thus referred to as \emph{no cooperation} (NC). It is worth noting that under this setup, the operation of each Tx-Rx link is equivalent to the point-to-point SWIPT system subject to time-varying co-channel interference which is studied in \cite{J_LZC:2013_a}. In this case, we assume that each Tx sends the  information signal only to its corresponding Rx with the maximum power $\Pmax$, over all the fading states, i.e., $x_k(n)=x_k^I(n)$, $k=1,2$, where $x_k^I(n)\sim\cC\cN(0,\Pmax)$. According to \eqref{Eq:Energy2} and \eqref{Eq:Rate}, the harvested power and achievable rate at fading state $\nu$ for Rx $k$ are expressed as
\begin{equation} \label{Eq:EnergyNC}
Q_k(\nu) = (1-\rho_k(\nu))\l(|h_{kk}(\nu)|^2  + |h_{k\bar{k}}(\nu)|^2 \r)\Pmax, \;k = 1,2,
 \end{equation}
\begin{equation} \label{Eq:RateNC}
R_k(\nu) = \rho_k(\nu)\log_2\l(1+\frac{|h_{kk}(\nu)|^2 \Pmax}{|h_{k\bar{k}}(\nu)|^2 \Pmax + \sigma_k^2}\r),  \; k = 1,2.
\end{equation}

It then follows that problem (P1) is reduced to the following problem with Rx mode switching variables only:
\begin{align}
\mathrm{(P3)}:~\mathop{\mathtt{Maximize}}_{\{\rho_1(\nu),\rho_2(\nu)\}} &~~  \E_\nu\l[R_1(\nu)+R_2(\nu)\r] \nn\\
\mathtt{subject\; to}&~~  \E_\nu\l[Q_k(\nu)\r] \geq \bar{Q}_k,  \; k=1,2\label{Ineq:P3}\\
&~~\rho_k(\nu) \in \{0,1\}, \;\forall \nu, k=1,2. \nn
\end{align}

Similar to (P1) and (P2), problem (P3) can be decoupled into subproblems each for one particular fading state and expressed as (by omitting the fading state $\nu$)
\begin{equation} \label{Eq:SubproblemNon}
\max_{\rho_1,\rho_2\in\{0,1\}} f_\nu^\NC(\rho_1,\rho_2),
\end{equation}
where 
\begin{equation} \label{Eq:Lagrangian_Non}
f\nu^\NC(\rho_1,\rho_2) = R_1 + R_2 + \mu_1 Q_1 + \mu_2 Q_2,
\end{equation}
with $\mu_1,\mu_2\geq 0$ denoting the dual variables associated with the constraints in \eqref{Ineq:P3} for $k=1,2$, respectively. Note that problem \eqref{Eq:SubproblemNon} can be solved by separately optimizing $\rho_1\in\{0,1\}$ and $\rho_2\in\{0,1\}$ by Rx $1$ and Rx $2$, respectively. According to \eqref{Eq:EnergyNC} and \eqref{Eq:RateNC}, the optimal solution to problem \eqref{Eq:SubproblemNon} is given by \cite{J_LZC:2013_a}
\begin{align*} 
\rho_k^\star = \l\{\begin{aligned}
&1,  \qquad\mbox{if}\; \log_2\l(1+\frac{|h_{kk}|^2\Pmax}{|h_{k\bar{k}}|^2 \Pmax + \sigma_k^2}\r) >\\ &\qquad\qquad\qquad\qquad \mu_k(|h_{kk}|^2  + |h_{k\bar{k}}|^2)\Pmax, \\
&0,  \qquad\mbox{otherwise},
\end{aligned}
\r.
\end{align*}
for $k=1,2$. Finally, (P3) can be solved by finding the optimal dual solution $(\mu_1^\star,\mu_2^\star)$, which can be determined by Rx $1$ and Rx $2$, respectively, by a simple bisection search.

Note that in this scheme, since there is no Tx-side cooperation, the network coordinator is not needed, which reduces the system complexity as compared to the other two cases of full and partial cooperation.

\section{Collaborative Transmission for SWIPT: $K$-User Case} \label{Section:Multi-user}
In this section, we study the general $K$-user SWIPT system with $K>2$. Similar to problem (P1), we can formulate the problem to maximize the average sum-capacity subject to the average harvested power constraint for each Rx. However, to avoid the high complexity of exhaustively searching for the optimal operation modes for all users as well as the corresponding signal splitting and precoding matrix for collaborative energy beamforming (as in Section~\ref{Section:OptimalSolution_Cooperative}) when $K$ becomes large, we propose a suboptimal scheme with lower complexity. This scheme is referred to as \emph{pairwise cooperation}, where we divide the $K$ users into $K/2$ groups (assuming $K$ is even), and then  apply the collaboration schemes obtained for the two-user case to the different groups. Furthermore, for a performance benchmark, we present a baseline scheme that is named as \emph{joint cooperation}, where all the users operate in either ID mode or EH mode synchronously  at each fading state, which is inspired by the principle of ergodic interference alignment introduced in \cite{J_NGJV:2012}. Note that for each scheme, the network coordinator is needed to coordinate the transmission of $K$ users.

\subsection{Pairwise Cooperation} \label{Subsection:Pairwise-cooperation}
First, we consider the pairwise cooperation based on the transmit cooperation schemes proposed in Section~\ref{Section:ProblemFormulation} for the two-user SWIPT system. For this scheme, we first divide the $K$ Tx-Rx pairs into $K/2$ groups with each group consisting  of two Tx-Rx pairs, and then apply the FC scheme in Section~\ref{Section:ProblemFormulation} to each group\footnote{For the case when $K$ is odd, we can group $K-1$ users with the proposed grouping scheme, where the remaining Tx-Rx link needs to perform mode switching independently without user pairing. }. 

We first address the key issue on how to group the users given channel conditions to guarantee good performance of the collaborative WIT and WET design. Although we can exhaustively search over the $K(K-1)/2$ possible grouping cases to obtain the one that leads to the best R-E performance, it should be noted that the complexity of such an exhaustive search is very high, i.e., $O(K^2)$ as $K$ becomes large. Thus, a more efficient and practical grouping algorithm is needed. However, intuitively there may be no straightforward solution to this problem, due to the conflicting goals between WIT versus WET. Specifically, for WIT, it is desirable to group the users to be far apart, in order to minimize the interference of both the intra-group and inter-group users; however, for WET, strong interference between the intra-group users is advantageous to achieve higher collaborative energy beamforming gains. In order to strike a balance between WIT and WET, we propose a simple grouping algorithm that  generally results in weak inter-group interference (for WIT), but strong intra-group interference (for collaborative WET). The main advantages of our proposed grouping algorithm is twofold. First, for collaborative WET, if the intra-group interference is strong, we can maximally exploit the collaborative energy beamforming gain within each group. Second, for efficient WIT, it is also expected that the strong intra-group interference could be avoided to certain extent by the opportunistic mode switching from mode $(\ID,\ID)$ to mode $(\ID,\EH)$ (or $(\EH,\ID)$), as well as the Tx-side power control in mode $(\ID,\ID)$.

Our proposed grouping algorithm is implemented as follows. First, we obtain the user indices with the largest average cross-link channel power over all the users in $\cK$, i.e., set $\{m,n\} = \arg\max_{\{k,l\}}\{\E_\nu[|h_{kl}(\nu)|^2]\}_{k,l\in\cK,k\neq l}$, and then group the $m$th and $n$th Tx-Rx pairs to be the first group. Next, we remove the grouped $m$th and $n$th pairs from the user set $\cK$, and repeat the same user selection until all $K$ users are grouped (assuming $K$ is even).

After the grouping, for simplicity, we assume that each group first ignores the inter-group interference to optimize their collaborative transmit signal design based on the FC scheme,  to avoid the complications due to the inter-group interference. However, after the intra-group FC scheme is designed, the actual achievable rate or harvested power for each grouped user pairs is computed by taking into account the inter-group interference for the sake of completeness.

\subsection{Joint Cooperation Based on Ergodic Interference Alignment} \label{Subsection:K-user Cooperation}
Next, we provide an alternative transmit cooperation design for the $K$-user SWIPT system based on the  \emph{ergodic interference alignment} (E-IA) \cite{J_NGJV:2012}, as a benchmark scheme for our proposed pairwise cooperation scheme. The main idea of E-IA is as follows. Given any fading state $\nu$, we define its \emph{complementary fading state} $\nu_C$ such that $h_{kl}(\nu) =  h_{kl}(\nu_C)$ if $k=l$ and  $h_{kl}(\nu) = - h_{kl}(\nu_C)$ if $k\neq l$, $k,l\in\cK$. It was shown in \cite{J_NGJV:2012} that we can obtain interference-free transmission of the $K$ links if all Txs send the same signals at fading state $\nu$ as well as at the  complementary fading state $\nu_C$ that appears in future. Assuming ideal channel quantization and no transmission delay constraint for the purpose of theoretical investigation, each Tx $k$ achieves the average rate $\frac{1}{2}\E_\nu[\log_2(1+2|h_{kk}(\nu)|^2 \Pmax/\sigma_k^2)]$, $k\in\cK$, for WIT with maximum transmit power $\Pmax$. It should be pointed out that perfect E-IA is difficult to achieve in practice due to the required infinitely long transmission delay to achieve half of the interference-free capacity; thus, we consider the scheme based on E-IA as a baseline scheme against which the performance of the proposed pairwise cooperation scheme is compared. 

Our proposed joint cooperation for the $K$-user SWIPT system based on the E-IA is then described as follows. At each fading state $\nu$, we assume that all users operate in either ID mode or EH mode. If ID mode is selected, as for E-IA, all Txs send independent information to their corresponding Rxs and will also send the same signals when the complementary fading state $\nu_C$ occurs in future. On the other hand, if EH mode is selected, all Txs cooperatively send energy signals to all Rxs via energy beamforming. At each fading state $\nu$, similar to \eqref{Eq:Mode}, we define an indicator function as 
\begin{equation*}
\rho^\mathsf{IA}(\nu) = \l\{\begin{aligned}
&1, \quad \mbox{ID mode is active}, \\
&0, \quad \mbox{EH mode is active}.
\end{aligned}
\r.
\end{equation*}
Then, for fading state $\nu$, the achievable rate of Rx $k\in\cK$ based on the E-IA is given by
\begin{align} \label{Eq:Rate_IA}
R_k^\mathsf{IA}(\nu) = \rho^\IA(\nu)\frac{1}{2}\log_2\l(1+\frac{2|h_{kk}(\nu)|^2\Pmax}{\sigma_k^2}\r), \; k\in \cK .
\end{align}
On the other hand, similar to \eqref{Eq:Energy2}, the harvested power at Rx $k$ at fading state $\nu$ is expressed as 
\begin{align} \label{Eq:Energy_IA}
Q_k^\IA(\nu) = (1-\rho^\IA(\nu)) \mv{h}_k(\nu)\mv{S}_E(\nu)\mv{h}_k^H(\nu), \; k\in \cK.
\end{align}

Similar to problem (P1), to characterize the resulting R-E performance of the above E-IA based joint cooperation scheme, we formulate the following problem.
\begin{align}
\mathrm{(P4)}:~\mathop{\mathtt{Maximize}}_{\{\rho^\mathsf{IA}(\nu),\mv{S}_E(\nu)\}} &~~ \sum_{k\in\cK}\E_\nu\l[R_k^\mathsf{IA}(\nu)\r]  \nn\\
\mathtt{subject\; to}&~~  \E_\nu\l[Q_k^\mathsf{IA}(\nu)\r] \geq \bar{Q}_k , \;k\in\cK \label{Ineq:P4_2} \\
&~~\rho^\mathsf{IA}(\nu) \in \{0,1\} , \;\forall \nu \nn\\
&~~\mv{S}_E(\nu)\in \mv{\cS}, \;\forall \nu, \nn
\end{align}
where $\mv{\cS}=\{\mv{S}_E(\nu): \mv{S}_E(\nu)\succeq\mv{0}, [\mv{S}_E(\nu)]_{k,k} \leq\Pmax, k\in\cK\}$ denotes the feasible set for $\mv{S}_E(\nu)$ subject to the peak transmit power constraint at each Tx. 

Similar to the two-user case, problem (P4) can be decoupled into subproblems each for one fading state and expressed as (by omitting the fading state $\nu$)
\begin{equation} \label{Eq:Subproblem3}
\max_{\rho^{\IA}\in\{0,1\},\mv{S}_E\in\mv{\cS}} f_\nu^\IA(\rho^{\IA},\mv{S}_E),
\end{equation}
where
\begin{equation} \label{Eq:f_IA}
f_\nu^\IA(\rho^{\IA},\mv{S}_E)= \sum_{k\in\cK} R_k^\IA + \sum_{k\in\cK} \mu_k Q_k^\IA, 
\end{equation}
with $\mu_k\geq 0$, $k\in\cK$, denoting the dual variable associated with the harvested power constraint in \eqref{Ineq:P4_2}. Problem \eqref{Eq:Subproblem3} can be solved by first obtaining the optimal $\mv{S}_E$, denoted by $\bar{\mv{S}}_E$, that maximizes $f_\nu^\IA(\rho^{\IA},\mv{S}_E)$ in \eqref{Eq:f_IA} for a given $\rho^\IA\in\{0,1\}$, and then finding $\rho^\IA\in\{0,1\}$ to maximize $f_\nu^\IA(\rho^{\IA},\bar{\mv{S}}_E)$. First, if $\rho^\IA = 1$, according to \eqref{Eq:Rate_IA}, it follows that with $\bar{\mv{S}}_E = {\rm Diag}(\Pmax,...,\Pmax)$,
\begin{equation} \label{Eq:Lagrange_IA_ID}
f_\nu^\IA(\rho^{\IA}=1,\bar{\mv{S}}_E) = \sum_{k\in\cK}\frac{1}{2} \log_2\l(1+\frac{2|h_{kk}|^2\Pmax}{\sigma_k^2}\r).
\end{equation}
Next, if $\rho^\IA = 0$, according to \eqref{Eq:Energy_IA}, problem \eqref{Eq:Subproblem3} is expressed as
\begin{align}
\label{P4.1}\mathrm{(P4.1)}:~\mathop{\mathtt{Maximize}}_{\mv{S}_E} &~~ \sum_{k\in\cK}\mu_k\mv{h}_k\mv{S}_E\mv{h}_k^H  \\
\mathtt{subject\; to}&~~  {\rm Tr}(\mv{I}_k \mv{S}_E) \leq \Pmax , \;k\in\cK \nn\\
&~~\mv{S}_E \succeq \mv{0}, \nn 
\end{align}
where $\mv{I}_k$ is defined such that $[\mv{I}_k]_{n,m}=1$ if $n=m=k$ and $0$ otherwise.  In fact, problem (P4.1) generalizes problem (P1.1) to the case with $K>2$, which is also a SDP. Although the closed-form solution of (P4.1) cannot be obtained with $K>2$ (unlike (P1.1) in the special case of $K=2$), we can apply existing software e.g., CVX  \cite{CVX} to solve this problem efficiently. 

According to \eqref{Eq:Lagrange_IA_ID} and \eqref{P4.1}, the optimal mode to problem \eqref{Eq:Subproblem3} is obtained  as
\begin{align*} 
\rho^{\IA\star} = \l\{\begin{aligned}
&1,  \qquad\mbox{if}\; \sum_{k\in\cK} \frac{1}{2}\log_2\l(1+\frac{2|h_{kk}|^2\Pmax}{\sigma_k^2}\r) > \\ &\qquad\qquad\qquad\qquad\sum_{k\in\cK}\mu_k\mv{h}_k\bar{\mv{S}}_E\mv{h}_k^H,\\
&0,  \qquad\mbox{otherwise}.
\end{aligned}
\r.
\end{align*}
Thus, given any set of dual variables $\{\mu_k\}$, $k\in\cK$, problem \eqref{Eq:Subproblem3} is efficiently solved. Finally, to find the optimal dual solution $\{\mu_k^\star\}$, $k\in\cK$, similarly as in Section~\ref{Section:OptimalSolution_Cooperative}, the ellipsoid method can be applied. Problem (P4) is thus solved. 

It is worth noting that unlike (P1.1) in the two-user case, in general the optimal solution to (P4.1) is not guaranteed to be of rank one with $K>2$, and thus more than one energy beams may need to be transmitted by the $K$ Txs for achieving the optimal WET, with which the comparison with pairwise cooperation (which adopts only a single energy-beam at all the Txs, as shown in Section~\ref{Section:ProblemFormulation}) may not be fair. To compensate this in some extent, the so-called \emph{randomization} techniques (see, e.g., \cite{J_SDL:2006} and references therein) can be employed to generate good suboptimal rank-one solutions based on the optimal solution of (P4.1) obtained without applying any rank constraint, for which the details are omitted for brevity.

\section{Simulation Results} \label{Section:Numerical}

In this section, we evaluate the performance of the proposed cooperation schemes for SWIPT by simulation. We set the peak transmit power as $\Pmax=20$ dBm or $0.1$ watt (W), the noise power as $\sigma_k^2 = -50$ dBm, and the Rx energy harvesting efficiency as $\zeta = 0.7$. In the following, we first show the results for the two-user SWIPT system, and then present the results for the general $K$-user SWIPT system.

\subsection{Two-User SWIPT System} \label{Subsection:Numerical_EachMode}
\begin{figure}
\centering
\includegraphics[width=8cm]{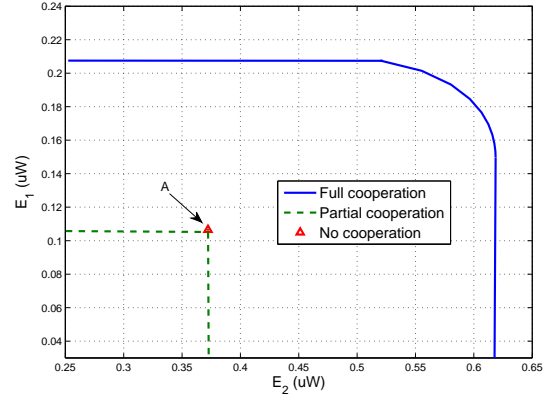}
\caption{Achievable E-E region in the AWGN channel for Mode $(\EH,\EH)$.} 
\label{Fig:E_E_Region}
\end{figure}
\begin{figure}
\centering
\includegraphics[width=8cm]{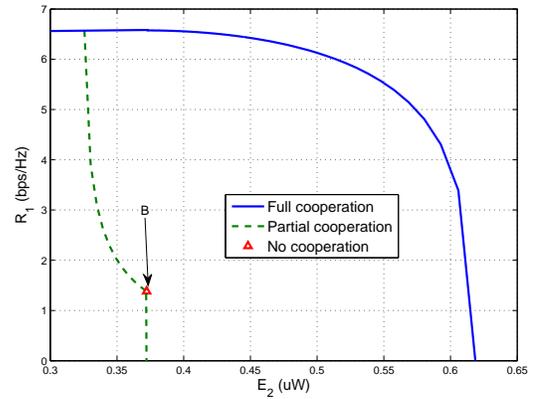}
\caption{Achievable R-E region in the AWGN channel for Mode $(\ID,\EH)$.}  \vspace{-10pt}
\label{Fig:R_E_Region}
\end{figure}

In this subsection, we consider a two-user SWIPT system. First, we show the performance gains for the two modes of $(\EH,\EH)$ and $(\ID,\EH)$ (or $(\EH,\ID)$) assuming an AWGN channel by the proposed full cooperation (FC) scheme with transmit energy beamforming and signal splitting, as compared to the existing partial cooperation (PC) and no cooperation (NC) schemes. By solving problems (P1.1) and (P1.2) with different weights, we obtain the resulting energy-energy (E-E) region and R-E region for the $(\EH,\EH)$ mode and $(\ID,\EH)$ mode, respectively, shown in Figs.~\ref{Fig:E_E_Region} and \ref{Fig:R_E_Region}, respectively. The channels are set as $h_{11} =0.0307e^{j1.7683}$, $h_{12} =0.0241e^{-j2.6973}$, $h_{21} = 0.0349e^{-j1.4011}$, and $h_{22} =0.0258e^{j2.8246}$, assuming an average $30$ dB of signal power attenuation for each pair of Tx and Rx. Notice that for the case of NC, only one single E-E or R-E point for the two links is achieved  (see point $A$ and point $B$ in Figs.~\ref{Fig:E_E_Region} and \ref{Fig:R_E_Region}, respectively). For $(\EH,\EH)$ mode or $(\ID,\EH)$ mode, it can be observed from Fig.~\ref{Fig:E_E_Region} or Fig.~\ref{Fig:R_E_Region} that the achievable E-E or R-E region by the proposed FC scheme remarkably outperforms that with PC and NC, thanks to the collaborative energy beamforming and the optimal signal splitting at the two Txs.

\begin{figure}
\centering
\subfigure[Case 1]{
\centering
\includegraphics[width=4.1cm]{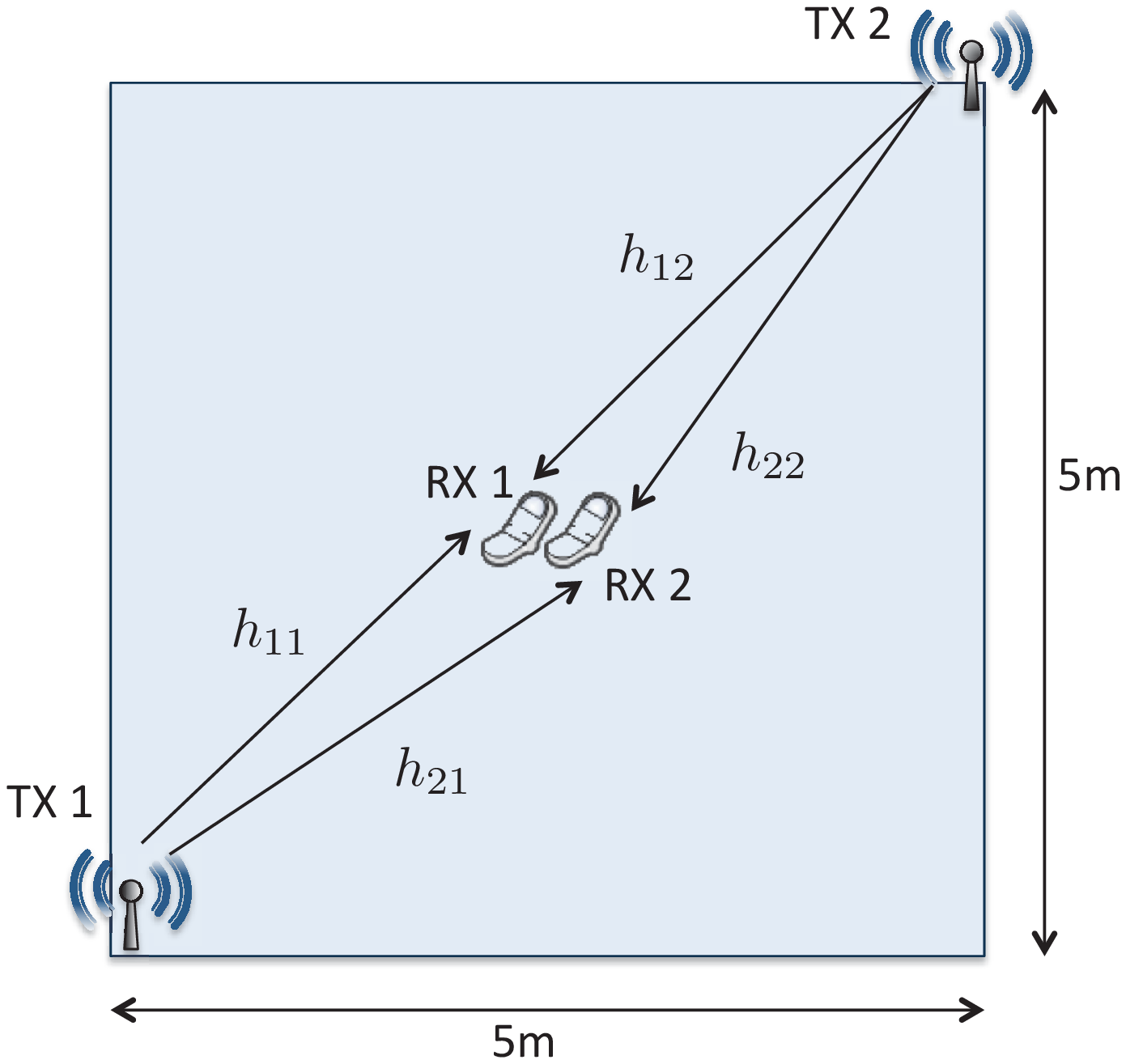}\label{Fig:Practical1}}
\subfigure[Case 2]{
\centering
\includegraphics[width=4.4cm]{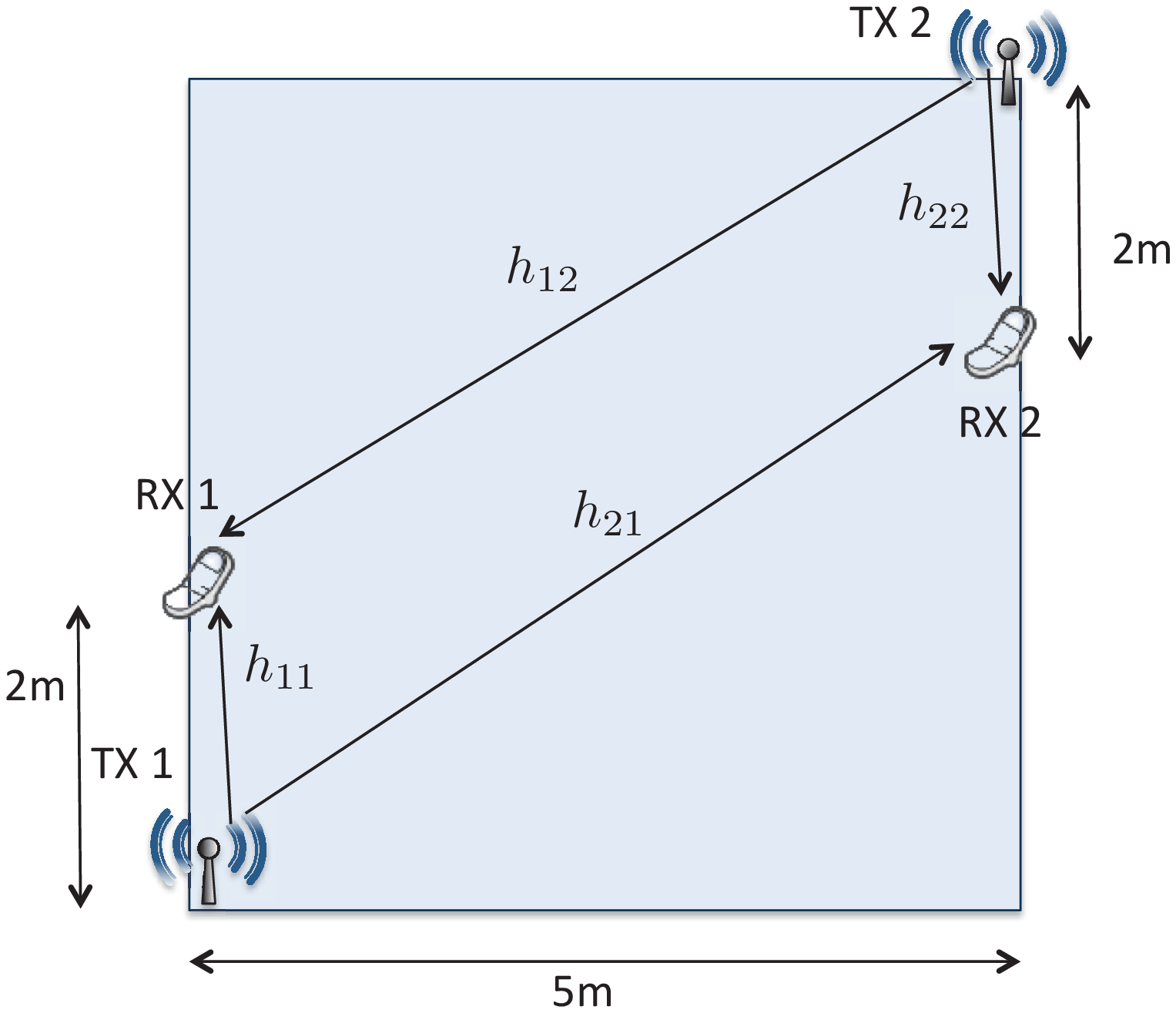}\label{Fig:Practical2}}
\caption{Simulation setup for Fig.~\ref{Fig:R_E_region_practical}. } \label{Fig:Practical} 
\end{figure}

\begin{figure}
\centering
\includegraphics[width=8cm]{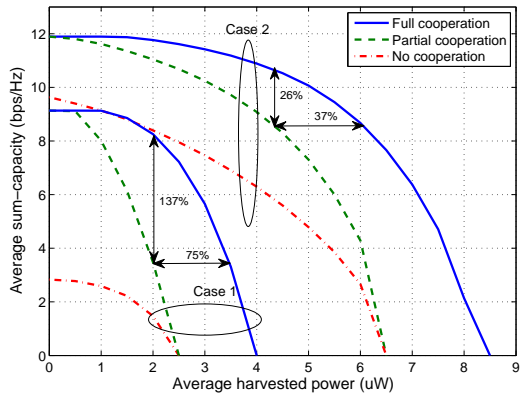}
\caption{R-E regions of the two-user SWIPT system in Rician fading channel.}  \vspace{-10pt}\label{Fig:R_E_region_practical}
\end{figure}

Inspired by the E-E and R-E performance gains in the AWGN channel, next, we show the simulation results on the achievable R-E region over flat-fading channels for different schemes, by solving problems (P1), (P2), and (P3) by setting different harvested power targets $\bar{Q}_1$ and $\bar{Q}_2$ for the two Rxs. In the simulation, we independently generate a sufficiently large number of fading states to approximate the continuous fading channel, and the theoretical expectation is obtained by sample average. For the simulation setup, it is assumed that Tx $1$ and Tx $2$ are located in two opposite corners in a $5$m$\times 5$m squared region, as shown in Fig.~\ref{Fig:Practical}. Under this setup, the line-of-sight (LoS) signal plays the dominant role, and thus Rician fading is used to model the channel, where for each fading state $\nu$ the complex channel $h_{kl}(\nu)$, $k,l\in\cK$, is defined as
{\small
\begin{equation} \label{Eq:ChannelModel}
h_{kl}(\nu) = \l(\sqrt{\frac{M}{M+1}}\hat{g} + \sqrt{\frac{1}{M+1}}g_{kl}(\nu)\r)\sqrt{c_0\l(\frac{r_{kl}}{r_0}\r)^{-\xi}},
\end{equation}} 
where $\hat{g}$ is the LoS deterministic component with $|\hat{g}|^2=1$; $g_{kl}(\nu)$ is a CSCG random variable with zero mean and unit variance denoting the short-term (Rayleigh) fading\footnote{For the short-term fading, we assume a scattering environment with moving scatters.}; $M$ is the Rician factor specifying the power ratio between the LoS and fading components in $g_{kl}(\nu)$, which is set as $M=3$; $c_0 = -20$ dB is a constant attenuation due to the path-loss at a reference distance $r_0 = 1$m at a carrier frequency assumed as $f_c=900$MHz; $\xi = 3$ is the path-loss exponent, and $r_{kl}$ is the distance between Tx $l$ and Rx $k$.
 For the purpose of exposition, we compare the following two cases with different Rx locations: In the first case, referred to as Case $1$, Rx $1$ and Rx $2$ are both located at the center of the region as shown in Fig.~\ref{Fig:Practical1}, in which both direct-link and interference-link have the same average received signal power for the two Rxs, while in the second case, referred to as Case $2$, Rx $1$ (Rx $2$) is located closer to Tx $1$ (Tx $2$) than Rx $2$ (Rx $1$) as shown in Fig.~\ref{Fig:Practical2}, in which the direct-link power is stronger than the interference-link power for each of the two links. 

Under the above setup, the achievable R-E regions are shown in Fig.~\ref{Fig:R_E_region_practical}. Note that we have set $\bar{Q}_1=\bar{Q}_2=\bar{Q}$ to plot the R-E regions. First, for both Cases $1$ and $2$, it is observed that the proposed FC achieves the best R-E trade-off as compared to the existing PC and NC schemes. Next, as observed from Fig.~\ref{Fig:R_E_region_practical}, the gain of FC is more substantial in Case $1$ than that in Case $2$ (due to stronger interference-link power). Finally, it is observed that  the R-E performance for each of the FC/PC/NC schemes is better in Case $2$ than that in Case $1$ (due to stronger direct-link power). The above results provide useful insights on how these schemes could perform in practical systems with different Tx and Rx locations. 

\subsection{Multiuser SWIPT System}

\begin{figure}
\centering
\subfigure[Grouping case 1]{
\centering
\includegraphics[width=4cm]{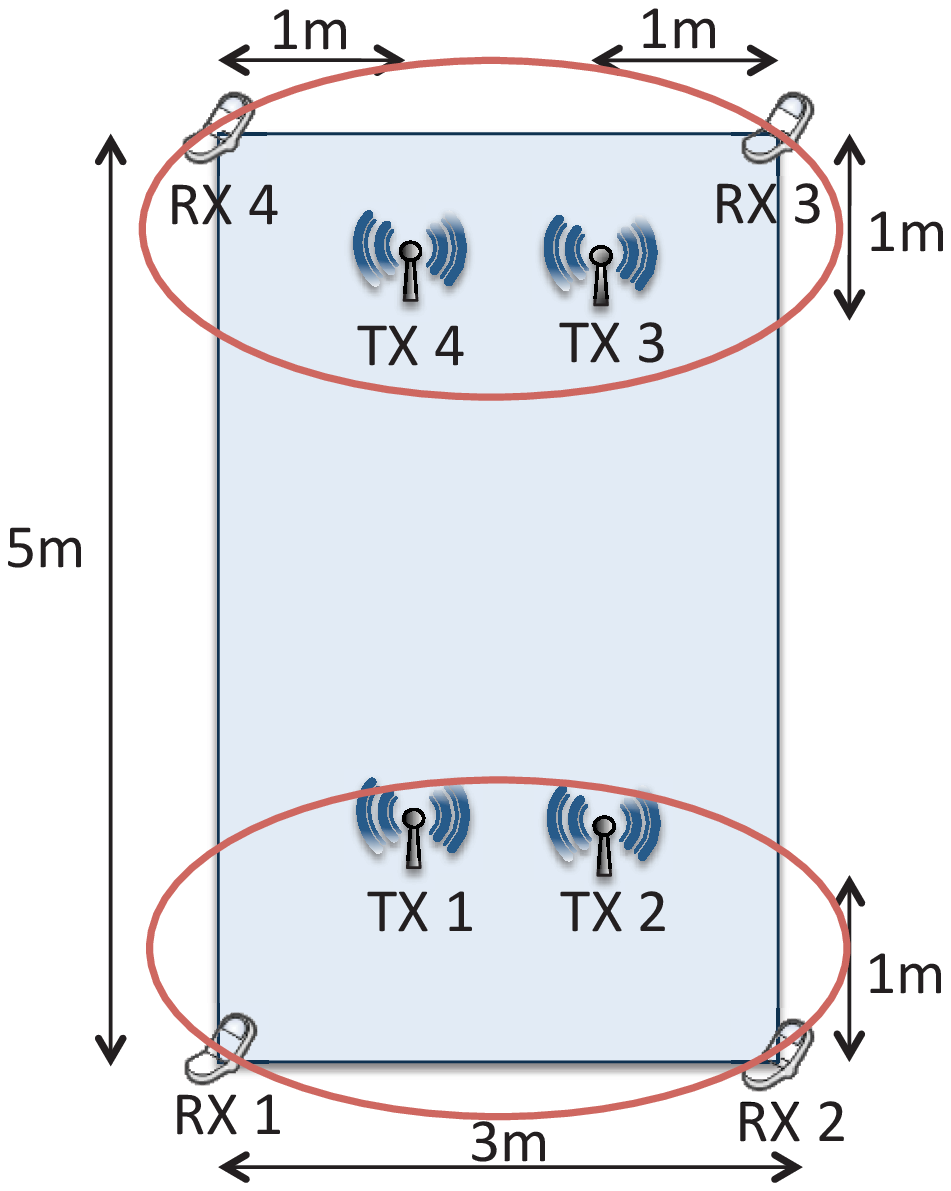}\label{Fig:Weaker}}
\subfigure[Grouping case 2]{
\centering
\includegraphics[width=4cm]{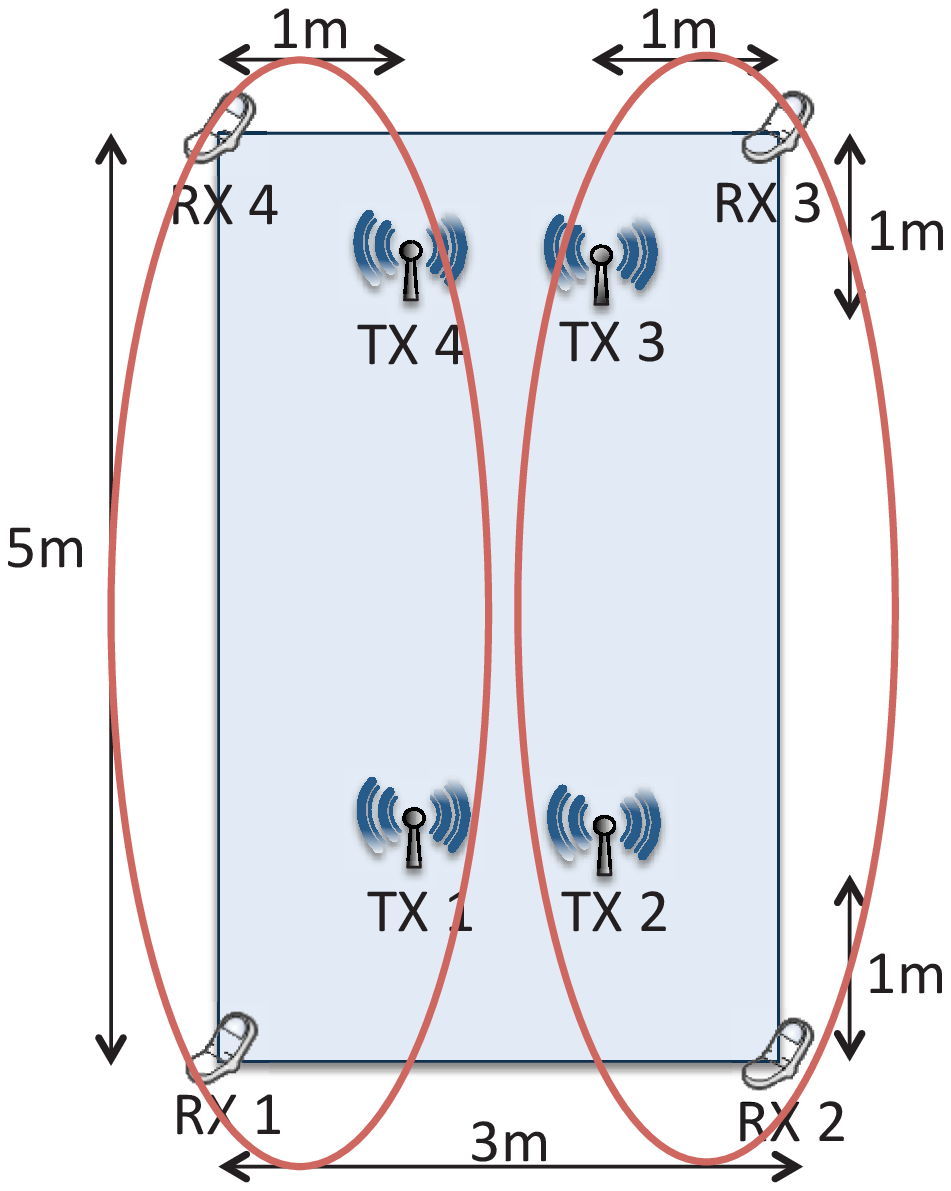}\label{Fig:Stronger}} 
\caption{Simulation setup for Fig.~\ref{Fig:Kuser_practical}.} \label{Fig:Multiuser} 
\end{figure}

\begin{figure}
\centering
\includegraphics[width=8cm]{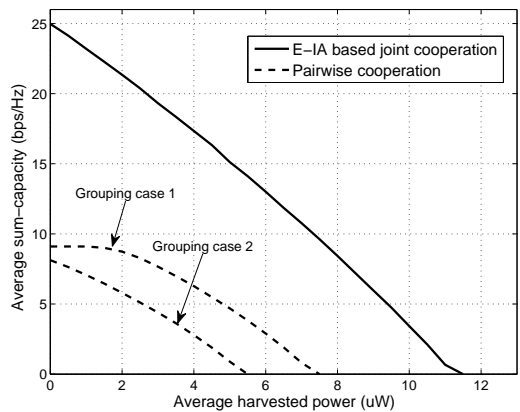}
\centering
\caption{R-E region for a four-user SWIPT system.} \vspace{-12pt}
\label{Fig:Kuser_practical}
\end{figure}

Next, we consider the $K$-user SWIPT system with $K>2$. We assume that $K=4$, and the channel model is similarly defined as in \eqref{Eq:ChannelModel}. The Tx/Rx locations are shown in ~Fig.~\ref{Fig:Multiuser}. Note that the grouping case $1$ in Fig.~\ref{Fig:Weaker} is given by our proposed grouping algorithm in Section~\ref{Subsection:Pairwise-cooperation}, which results in weaker inter-group interference but stronger intra-group interference as compared to  the grouping case $2$ in Fig.~\ref{Fig:Stronger}. Also note that the result on the E-IA based joint cooperation is based on suboptimal energy beamforming scheme with one single energy beam, obtained by randomization technique in Section~\ref{Subsection:K-user Cooperation}. First, as observed from Fig.~\ref{Fig:Kuser_practical}, under this particular setup, the E-IA based joint cooperation achieves better R-E trade-off than that of pairwise cooperation, thanks to the Txs' joint collaborative energy beamforming and E-IA based DoF (degrees-of-freedom) optimal WIT under the high-SNR regime considered here due to short-range communication\footnote{It should be noted from \cite{J_NGJV:2012} that E-IA requires symmetric phase distribution (e.g., uniform distribution) of the channels to achieve half of the interference-free rate as given in \eqref{Eq:Rate_IA}; however, the Rician channel model considered here does not satisfy such requirement due to the deterministic LoS component. As a result, the rate obtained from this simulation is not achievable in general and thus only serves as a performance upper bound. }. Next, it is observed that for the pairwise cooperation, the grouping case $1$ in Fig.~\ref{Fig:Weaker} by our proposed grouping algorithm performs better than the grouping case $2$ in Fig.~\ref{Fig:Stronger}. In fact, it has been verified by exhaustive search that under this setup the grouping case 1 is indeed optimal.

\section{Conclusions} \label{Section:Conclusion}
This paper has studied SWIPT under a multiuser interference channel setup. A new transmit scheme is proposed, namely signal splitting, to facilitate collaborative transmit energy beamforming. For the two-user case, we derive the optimal receiver mode switching rule and corresponding transmit optimization to achieve various  R-E trade-offs over fading channels. By comparing the two existing schemes with partial/no transmit cooperation, we show by simulation that there are notable R-E performance gains in SWIPT achieved by the proposed full cooperation scheme. Finally, the general case of multiuser SWIPT system is investigated  and two cooperation schemes are proposed, which are users' grouping-based pairwise cooperation and ergodic interference alignment based joint cooperation, respectively.

As for future work, it will be interesting to extend the results to the MIMO multiuser SWIPT system with multiple antennas at the transmitters and receivers, where spatial-domain interference alignment can be jointly designed with collaborative energy beamforming to optimize the R-E performance. 

\appendices 

\section{Proof of Proposition~\ref{Proposition:P1}} \label{Proof:Proposition:P1}
Since problem (P1.1) is a SDP, it is convex. Furthermore, it can be easily checked that this problem satisfies the Slater's condition. Thus, strong duality holds for (P1.1) and its dual problem \cite{B_BV:2004}. Similar to (P1), we can apply the Lagrange duality method to solve (P1.1). The Lagrangian of (P1.1) is formulated as 
{\small
\begin{align} \label{Eq:LagrangianP3.1}
\cL(\mv{S}_E,\lambda_1,\lambda_2) 
&= \mu_1 \mv{h}_1\mv{S}_E\mv{h}_1^H + \mu_2 \mv{h}_2\mv{S}_E\mv{h}_2^H  -\nn\\ 
& \quad \lambda_1({\rm Tr}(\mv{I}_1\mv{S}_E) - \Pmax)  - \lambda_2({\rm Tr}(\mv{I}_2\mv{S}_E) - \Pmax) \nn\\
&= {\rm Tr}((\mu_1 \mv{h}_1^H\mv{h}_1 + \mu_2 \mv{h}_2^H\mv{h}_2-\lambda_1 \mv{I}_1 - \lambda_2 \mv{I}_2)\mv{S}_E) + \nn\\ 
& \qquad(\lambda_1+\lambda_2)\Pmax, 
\end{align} }
where $\lambda_1$, $\lambda_2\geq 0$ are the dual variables  associated with the constraints in \eqref{Ineq:P3.1_1} and  \eqref{Ineq:P3.1_2}, respectively. 
The Lagrange dual function of (P1.1) is then given by
{\small
\begin{align} 
&u(\lambda_1,\lambda_2)  =  \max_{\mv{S}_E\succeq\mv{0}}\cL(\mv{S}_E,\lambda_1,\lambda_2) \nn\\
&= \l\{\begin{aligned}& +\infty,  &\mbox{if} \,\,  \mu_1 \mv{h}_1^H\mv{h}_1 + \mu_2 \mv{h}_2^H\mv{h}_2 - \lambda_1 \mv{I}_1 - \lambda_2 \mv{I}_2  \succ  \mv{0}, \\ 
& (\lambda_1  + \lambda_2) \Pmax,  &\mbox{if} \,\, \mu_1 \mv{h}_1^H\mv{h}_1 + \mu_2 \mv{h}_2^H\mv{h}_2-\lambda_1 \mv{I}_1 - \lambda_2 \mv{I}_2  \preceq\mv{0}. \label{Dual3.1}
\end{aligned}
\r. 
\end{align} }
As a result, the dual problem of (P1.1) is given by
\begin{align}
\mathrm{(D1.1)}:~\mathop{\mathtt{Minimize}}_{\lambda_1,\lambda_2\geq 0} &~~ (\lambda_1 + \lambda_2)\Pmax  \nn \\
\mathtt{subject \; to}&~~ \mu_1 \mv{h}_1^H\mv{h}_1 + \mu_2 \mv{h}_2^H\mv{h}_2-\lambda_1 \mv{I}_1 - \lambda_2 \mv{I}_2   \nn\\  &\preceq \mv{0}. \label{Ineq:Semidefinite2}
\end{align}
To solve (D1.1), we re-express \eqref{Ineq:Semidefinite2} as
\begin{equation} \label{Eq:A}
\l[\begin{array}{cc} ||\tilde{\mv{h}}_{1}||^2 - \lambda_1 & \tilde{\mv{h}}_1^H \tilde{\mv{h}}_2 \\ \tilde{\mv{h}}_2^H \tilde{\mv{h}}_ 1& ||\tilde{\mv{h}}_{2}||^2 - \lambda_2  \end{array}\r]\preceq \mv{0} ,
\end{equation}
where we have defined $\tilde{\mv{h}}_{k} = [\sqrt{\mu_1} h_{1k}, \sqrt{\mu_2} h_{2k}]^T$, $k=1,2$. From the theory of Schur complement \cite{B_BV:2004}, the condition in  \eqref{Eq:A} holds if and only if
\begin{align}
||\tilde{\mv{h}}_1||^2 - \lambda_1 &\leq 0, \nn\\
||\tilde{\mv{h}}_2||^2  - \lambda_2 - \frac{|\tilde{\mv{h}}_1^H \tilde{\mv{h}}_2|^2}{||\tilde{\mv{h}}_{1}||^2 - \lambda_1} &\leq 0. \label{Ineq:Schur2}
\end{align} 
It then follows from \eqref{Ineq:Schur2} that
\begin{equation} \label{Ineq:Schur3}
\lambda_2 \geq \frac{|\tilde{\mv{h}}_1^H \tilde{\mv{h}}_2|^2}{\lambda_1 - ||\tilde{\mv{h}}_{1}||^2} +  ||\tilde{\mv{h}}_2||^2. 
\end{equation}
Adding $\lambda_1$ to both sides of \eqref{Ineq:Schur3} yields
\begin{align}
\lambda_1 + \lambda_2 &\geq \frac{|\tilde{\mv{h}}_1^H \tilde{\mv{h}}_2|^2}{\lambda_1 - ||\tilde{\mv{h}}_1||^2} + \lambda_1 +   ||\tilde{\mv{h}}_2||^2 \nn\\
& = \frac{|\tilde{\mv{h}}_1^H \tilde{\mv{h}}_2|^2}{\lambda_1 - ||\tilde{\mv{h}}_1||^2} + \lambda_1 -||\tilde{\mv{h}}_1||^2 +  ||\tilde{\mv{h}}_1||^2  + ||\tilde{\mv{h}}_2||^2 \nn\\
&\geq 2|\tilde{\mv{h}}_1^H \tilde{\mv{h}}_2|  + ||\tilde{\mv{h}}_1||^2  + ||\tilde{\mv{h}}_2||^2, \label{Ineq:DualOpt}
\end{align}
where \eqref{Ineq:DualOpt} comes from the inequality of arithmetic and geometric means, and thus the equality in \eqref{Ineq:DualOpt} holds if and only if $\lambda_1$ equals $||\tilde{\mv{h}}_1||^2 + |\tilde{\mv{h}}_1^H \tilde{\mv{h}}_2|$. Hence, we conclude that the optimal solution of (D1.1) is given by $\lambda_1^\star = ||\tilde{\mv{h}}_1||^2 + |\tilde{\mv{h}}_1^H \tilde{\mv{h}}_2|$ and $\lambda_2^\star = ||\tilde{\mv{h}}_2||^2 + |\tilde{\mv{h}}_1^H \tilde{\mv{h}}_2|$. 

Next, based on the above result, we proceed to derive the optimal $\mv{S}_E$. Since (P1.1) is convex,  the optimal primal and dual solutions, denoted by $\bar{\mv{S}}_E$ and $\lambda_1^\star,\lambda_2^\star$,  must satisfy the following complementary slackness conditions: 
\begin{align} \label{KKTk}
&\lambda_k^\star({\rm Tr}(\mv{I}_k\bar{\mv{S}}_E) - \Pmax)  = 0,\; k=1,2. 
\end{align}
Since it has been shown above that $\lambda_k^\star$'s are strictly positive, it follows from \eqref{KKTk} that 
\begin{align*} 
{\rm Tr}(\mv{I}_k\bar{\mv{S}}_E) = \Pmax, \; k=1,2.
\end{align*}
Therefore, since $\bar{\mv{S}}_E$ is Hermitian, it can be expressed  as 
\begin{equation} \label{Eq:S_Hermitian}
\bar{\mv{S}}_E=\l[\begin{array}{cc} \Pmax & x \\ x^* & \Pmax  \end{array}\r].
\end{equation}
Moreover, it can be inferred from \eqref{Eq:LagrangianP3.1} and \eqref{Dual3.1} that
\begin{equation}
(\mu_1 \mv{h}_1^H\mv{h}_1 + \mu_2 \mv{h}_2^H\mv{h}_2-\lambda_1 \mv{I}_1 - \lambda_2 \mv{I}_2)\bar{\mv{S}}_E = \mv{0}. \label{KKT3}
\end{equation}
By solving the two linear equations from \eqref{KKT3} with $\bar{\mv{S}}_E$ given in \eqref{Eq:S_Hermitian}, $\bar{\mv{S}}_E$ can be obtained as
\begin{equation*} 
\bar{\mv{S}}_E=\Pmax\l[\begin{array}{cc} 1 & \alpha  \\ \frac{1}{\alpha}  & 1  \end{array}\r],
\end{equation*}
where $\alpha = \frac{\tilde{\mv{h}}_1^H \tilde{\mv{h}}_2}{\l|\tilde{\mv{h}}_1^H \tilde{\mv{h}}_2\r|} = \frac{\l|\tilde{\mv{h}}_1^H \tilde{\mv{h}}_2\r|}{\tilde{\mv{h}}_2^H \tilde{\mv{h}}_1}$. The proof is thus completed.

\bibliographystyle{ieeetr}

\begin{thebibliography}{1}

\bibitem{J_ZH:2013}
R.~Zhang and C.~K. Ho, ``{MIMO broadcasting for simultaneous wireless
  information and power transfer},'' {\em IEEE Trans. Wireless Commun.},
  vol.~12, pp.~1989--2001, May 2013.

\bibitem{A_ZZH}
X.~Zhou, R.~Zhang, and C.~K. Ho, ``{Wireless information and power transfer:
  architecture design and rate-energy tradeoff},'' {\em IEEE Trans. Commun.},
  vol.~61, pp.~4757--4767, Nov. 2013.

\bibitem{J_LZC:2013_a}
L.~Liu, R.~Zhang, and K.-C. Chua, ``{Wireless information transfer with
  opportunistic energy harvesting},'' {\em IEEE Trans. Wireless Commun.},
  vol.~12, pp.~288--300, Jan. 2013.

\bibitem{J_LZC:2013_b}
L.~Liu, R.~Zhang, and K.-C. Chua, ``{Wireless information and power transfer: a
  dynamic power splitting approach},'' {\em IEEE Trans. Commun.}, vol.~61,
  pp.~3990--4001, Sept. 2013.

\bibitem{J_FO:2012}
A.~M. Fouladgar and O.~Simeone, ``{On the transfer of information and energy in
  multi-user systems},'' {\em IEEE Commun. Letters}, vol.~16, pp.~1733--1736,
  Nov. 2012.

\bibitem{J_NZDK:2013}
A.~A. Nasir, X.~Zhou, S.~Durrani, and R.~A. Kennedy, ``{Relaying protocols for
  wireless energy harvesting and information processing},'' {\em IEEE Trans.
  Wireless Commun.}, vol.~12, pp.~3622--3636, July 2013.

\bibitem{J_DPEP:2014}
Z.~Ding, S.~M. Perlaza, I.~Esnaola, and H.~V. Poor, ``{Power allocation
  strategies in energy harvesting wireless cooperative networks},'' {\em IEEE
  Trans. Wireless Commun.}, vol.~13, pp.~846--860, Feb. 2014.

\bibitem{J_PFS:2013}
P.~Popovski, A.~M. Fouladgar, and O.~Simeone, ``{Interactive joint transfer of
  energy and information},'' {\em IEEE Trans. Commun.}, vol.~61,
  pp.~2086--2097, May 2013.

\bibitem{C_XLZ:2013}
J.~Xu, L.~Liu, and R.~Zhang, ``{Multiuser MISO beamforming for simultaneous
  wireless information and power transfer},'' {\em IEEE Trans. Sig. Process.},
  vol.~62, pp.~4798--4810, Sept. 2014.

\bibitem{A_SLXZ}
Q.~Shi, L.~Liu, W.~Xu, and R.~Zhang, ``{Joint transmit beamforming and receive
  power splitting for MISO SWIPT systems},'' {\em IEEE Trans. Wireless
  Commun.}, vol.~13, pp.~3269--3280, June 2014.

\bibitem{J_HL:2013}
K.~Huang and E.~G. Larsson, ``{Simultaneous information and power transfer for
  broadband wireless systems},'' {\em IEEE Trans. Sig. Process.}, vol.~61,
  pp.~5972--5986, Dec. 2013.

\bibitem{A_ZZH_b}
X.~Zhou, R.~Zhang, and C.~K. Ho, ``{Wireless information and power transfer in
  multiuser OFDM systems},'' {\em IEEE Trans. Wireless Commun.}, vol.~13,
  pp.~2282--2294, Apr. 2014.

\bibitem{A_NLS}
D.~W.~K. Ng, E.~S. Lo, and R.~Schober, ``{Wireless information and power
  transfer: energy efficiency optimization in OFDMA systems},'' {\em IEEE
  Trans. Wireless Commun.}, vol.~12, pp.~6352--6370, Dec. 2013.

\bibitem{J_LZC:2014}
L.~Liu, R.~Zhang, and K.-C. Chua, ``{Secrecy wireless information and power
  transfer with MISO beamforming},'' {\em IEEE Trans. Sig. Process.}, vol.~62,
  pp.~1850--1863, Apr. 2014.

\bibitem{J_NLS:2014}
D.~W.~K. Ng, E.~S. Lo, and R.~Schober, ``{Robust beamforming for secure
  communication in systems with wireless information and power transfer},''
  {\em IEEE Trans. Wireless Commun.}, vol.~13, pp.~4599--4615, Aug. 2014.

\bibitem{C_SLC:2012}
C.~Shen, W.-C. Li, and T.-H. Chang, ``{Simultaneous information and energy
  transfer: a two-user MISO interference channel case},'' in {\em Proc. IEEE
  Global Commun. Conf. (GLOBECOM)}, pp.~3862--3867, 2012.

\bibitem{C_TKO:2013}
S.~Timotheou, I.~Krikidis, and B.~Ottersten, ``{MISO interference channel with
  QoS and RF energy harvesting constraints},'' in {\em Proc. IEEE Int. Conf.
  Commun. (ICC)}, pp.~4191--4196, 2013.

\bibitem{J_PC:2013}
J.~Park and B.~Clerckx, ``{Joint wireless information and energy transfer in a
  two-user MIMO interference channel},'' {\em IEEE Trans. Wireless Commun.},
  vol.~12, pp.~4210--4221, Aug. 2013.

\bibitem{J_NGJV:2012}
B.~Nazer, M.~Gastpar, S.~A. Jafar, and S.~Vishwanath, ``{Ergodic interference
  alignment},'' {\em IEEE Trans. Inf. Theory}, vol.~58, pp.~6355--6371, Oct.
  2012.

\bibitem{J_YL:2006}
W.~Yu and R.~Lui, ``{Dual methods for nonconvex spectrum optimization of
  multicarrier systems},'' {\em IEEE Trans. Commun.}, vol.~54, pp.~1310--1322,
  July 2006.

\bibitem{CVX}
M.~Grant and S.~Boyd, {\em {CVX: Matlab software for disciplined convex
  programming, version 1.21}}.
\newblock http://cvxr.com/cvx/, 2911.

\bibitem{B_BV:2004}
S.~Boyd and L.~Vandenberghe, {\em {Convex Optimization}}.
\newblock Cambridge University Press, 2004.

\bibitem{J_GGOK:2008}
A.~Gjendemsjo, D.~Gesbert, G.~E. Oien, and S.~G. Kiani, ``{Binary power control
  for sum rate maximization over multiple interfering links},'' {\em IEEE
  Trans. Wireless Commun.}, vol.~7, pp.~3164--3173, Aug. 2008.

\bibitem{J_SDL:2006}
N.~D. Sidiropoulos, T.~N. Davidson, and Z.-Q. Luo, ``{Transmit beamforming for
  physical-layer multicasting},'' {\em IEEE Trans. Sig. Process.}, vol.~54,
  pp.~2239--2251, June 2006.


\end{thebibliography}

\end{document}